\documentclass[3p, 12pt, sort&compress]{elsarticle}
\makeatletter
\def\ps@pprintTitle{%
 \let\@oddhead\@empty
 \let\@evenhead\@empty
 \def\@oddfoot{}%
 \let\@evenfoot\@oddfoot}
\makeatother
\usepackage{float} 
\usepackage{graphicx}
\usepackage{amsmath}
\usepackage{amssymb}
\usepackage{bm}
\usepackage{bbm}
\usepackage{framed}
\usepackage{subfigure}
\usepackage{hyperref}
\usepackage{nomencl}
\usepackage{color}
\usepackage{textcomp}
\usepackage{graphics}
\usepackage{epsfig}
\usepackage{bm}
\usepackage{epstopdf}
\usepackage{amsthm}
\usepackage{float}
\usepackage{latexsym,amsfonts}
\usepackage{url}
\usepackage{longtable}
\usepackage[figuresright]{rotating}
\usepackage{listings}
\usepackage{algorithm}
\usepackage{algpseudocode}
\newlength\figureheight 
\newlength\figurewidth
\usepackage{scalefnt}
\usepackage{pgfplots}
\pgfplotsset{compat=newest}
\pgfplotsset{plot coordinates/math parser=false}

\usepackage{makeidx}
\makeindex

\theoremstyle{definition}

\newtheorem{proposition}{Proposition}

\newtheorem{remark}{Remark}

\begin{document}

\begin{frontmatter}
\title{Efficient Monte Carlo Simulation of the Left Tail of Positive Gaussian Quadratic Forms}
\author[rvt]{Chaouki Ben Issaid}
\ead{chaouki.benissaid@kaust.edu.sa}

\author[rvt]{Mohamed-Slim Alouini}
\ead{slim.alouini@kaust.edu.sa}

\author[rvt,fct]{Ra\'ul Tempone}
\ead{raul.tempone@kaust.edu.sa}

\address[rvt]{King Abdullah University of Science and Technology (KAUST), Computer, Electrical and Mathematical Science and Engineering (CEMSE) Division, Thuwal 23955-6900, Saudi Arabia}
\address[fct]{Alexander von Humboldt Professor in Mathematics for Uncertainty Quantification, RWTH Aachen University, 52062 Aachen, Germany}

\begin{abstract}
Estimating the left tail of quadratic forms in Gaussian random vectors is of major practical importance in many applications. In this paper, we propose an efficient and robust importance sampling estimator that is endowed with the bounded relative error property. This property significantly reduces the number of simulation runs required by the proposed estimator compared to naive Monte Carlo. Thus, our importance sampling estimator is especially useful when the probability of interest is very small. Selected simulation results are presented to illustrate the efficiency of our estimator compared to naive Monte Carlo in both central and non-central cases, as well as both real and complex settings.
\end{abstract}


\end{frontmatter}

\section{Introduction}
Quadratic forms can appear when the effect of inequality between errors in terms of variance and correlation is examined in a two-way analysis of variance \cite{box1954}, when the constrained least-squares estimator is studied \cite{HSUAN1985}, and during statistical hypothesis testing. Many test statistics,such as the test statistics in covariance structure analysis \cite{Shapiro1983}, and the general likelihood ratio statistic \cite{Vuong1989}, can be expressed in terms of quadratic forms. These tests have a wide range of applicability. For instance, the multilocus association test for the genetic dissection of complex diseases in genetic studies \cite{Tang2010}, the spectral analysis of the Wishart ensemble or counting string vacua in fields such as string theory \cite{Bausch2013}, and non-coherent detection \cite{Divsalar1990}-\cite{Raphaeli1996} and combining diversity \cite[Chap. 14]{Proakis1994} in communication theory.

Gurland investigated the distribution of quadratic forms and ratios of quadratic forms \cite{gurland1953, gurland1955}. The author presented these distributions in terms of an infinite sum involving Laguerre Polynomials provided that the semi-moments are known. However, Gurland's expression of the coefficients is not very suitable for computation, and the estimate of the truncation error is given under certain assumptions. Other works, such as Ruben \cite{ruben1960, ruben1962} and Shah \cite{shah1961, shah1963}, presented the distribution in terms of MacLaurin series or $\chi^2$ densities. As noted by Shah \cite{shah1963}, most of these expansions are not practical, or some of them do not provide an estimate of the truncation errors. Kotz \textit{et al.} unified these approaches in their works \cite{kotz1967a, kotz1967b} and derived a series representation for both the central and non-central cases.  

In general, the exact distribution of a linear combination of independent chi-square variates is a challenging task. In fact, various approximations, have been proposed in the literature, for example, Imhof \cite{imhof1961}, Davies \cite{davies1963}, and Solomon and Stephens \cite{solomon1977}. Based on the work of Imhof \cite{imhof1961}, Davies \cite{davies1963} presented a numerical approach to inverting the characteristic function of a real random variable with the aim of computing its left tail. The author showed that the method accurately produces the distribution of a central chi-squared random variable for various numbers of degrees of freedom. Rice (1980) \cite{Rice1980} presented a generalization of this approach, including two numerical integration methods of inverting the characteristic function. As Bausch commented \cite{Bausch2013}, when the probability of interest is very small, most of the existing methods fail to give an accurate result.

It is widely known that when the number of simulation runs of the model is limited and the probability is small (i.e., the occurrence of an event is rare), the naive Monte Carlo (MC) estimator is expensive. As an alternative, we propose in this work an efficient importance sampling (IS) estimator to estimate the probability of interest. The IS method is often used to estimate rare events probabilities. By modifying the dynamics of the simulations, i.e. by introducing a new change of probability measure, the rare event is no longer rare. Therefore, the IS method aims to reduce the number of required simulation runs given a certain confidence interval. However, proposing a poor choice of the new PDF will lead to a large likelihood ratio and, thus, to an estimator with a variance greater than the original MC estimator. That being said, we note that constructing a good biased distribution is the core of importance sampling. To the best of our knowledge, only two works proposing IS schemes for the purpose of computing tails of quadratic forms in Gaussian random vectors have been previously published \cite{Shi2016, Shi2018}. In those works, the authors were interested in the right tail and implemented IS combined with the cross-entropy method. However, the authors did not provide any efficiency analysis of their estimators. In this work, we are interested in estimating the left tail of positive quadratic forms in Gaussian random vectors using IS. We also show that the proposed IS scheme is endowed with the bounded relative error property.

The rest of this paper is organized as follows. First, we briefly describe the problem setting, and provide a lower bound for the probability of interest. In Section \ref{section2}, we prove the efficiency of our proposed estimator. After reviewing some basic notions of IS in Section \ref{section3}, we present our approach to estimating the probability of interest in Section \ref{section4}. We show some simulation results related to selected example of applications prior to concluding the paper. In each example, we compare the computational efficiency of our approach to that of naive MC.
\section{Problem Setting}\label{section2}
Let $X = (X_1,\dots, X_N)^{T}$ be a Gaussian random vector with PDF
\begin{align}
f_X(X) = \frac{\exp\left(-\frac{1}{2} (X-\mu)^{T} \Sigma_X^{-1} (X-\mu)\right)}{\sqrt{(2\pi)^{N} |\Sigma_X|}} ,
\end{align}
where $\mu$ is the mean vector, $\Sigma_X$ is the $(N \times N)$ covariance matrix, assumed to be strictly positive definite, and $|\cdot|$ represents the determinant of a matrix. For a given positive definite matrix $\Sigma \in \mathbb{R}^{N\times N}$ and a threshold $\gamma_0 > 0$, we aim to introduce an efficient IS scheme for computing the left tail of the quadratic form $X^{T} \Sigma X$, i.e.,
\begin{align}\label{P}
P = \mathbb{P}(X^{T} \Sigma X \leq \gamma_0),
\end{align}
as $\gamma_0 \rightarrow 0$. 

Before giving a lower bound on the probability $P$, we re-write its expression more conveniently. First, we write $X = \mu + \Sigma^{\frac{1}{2}} Y$, where $Y$ is a standard Gaussian vector. Then, we have \cite[Ch. 3]{QPM}
\begin{align}\label{qqd}
\nonumber &X^{T} \Sigma X = \left(\mu + \Sigma^{\frac{1}{2}} Y\right)^{T}  \Sigma  \left(\mu + \Sigma^{\frac{1}{2}} Y\right) \\
&= \left(Y + \Sigma^{-\frac{1}{2}} \mu \right)^{T} A \left(Y + \Sigma^{-\frac{1}{2}} \mu \right),
\end{align}
where $A = \Sigma_X^{\frac{1}{2}} \Sigma \Sigma_X^{\frac{1}{2}}$.\\
Note that $A$ is a symmetric matrix, thus using the spectral theorem, there exists an orthogonal matrix $Q$ and a diagonal matrix $\Lambda = diag(\lambda_1,\dots,\lambda_N)$ such that $A = Q^T \Lambda Q$. Now, let $W \neq 0$, then, we have $W^T A W = W^T \Sigma_X^{\frac{1}{2}} \Sigma \Sigma_X^{\frac{1}{2}} W > 0$, since the matrix $\Sigma$ is a positive definite matrix. Therefore, the eigenvalues $\{\lambda_i\}_{i=1}^{N}$ are non-negative. Going back to (\ref{qqd}) and replacing $A$ by its spectral decomposition, we get
\begin{align}
\nonumber &X^{T} \Sigma X = \left( Y + \Sigma^{-\frac{1}{2}} \mu \right)^{T} Q^T \Lambda Q  \left( Y + \Sigma^{-\frac{1}{2}} \mu \right)\\
& = (Z + \alpha)^{T} \Lambda (Z + \alpha),
\end{align}
where $Z = Q Y $ and $\alpha = Q \Sigma^{-\frac{1}{2}} \mu$. We note that $Z$ is a Gaussian random vector with zero mean and a covariance matrix $I$ because the matrix $Q$ is orthogonal, i.e., $Q^{T} Q = Q Q^{T} = I$, where $I$ is the identity matrix. Now, we return to the probability of interest $P$ and re-write it as \cite[Ch. 3]{QPM}
\begin{align}
P = \mathbb{P}\left(S_N = \sum_{i=1}^{N}{\lambda_i (Z_i + \alpha_i)^2} \leq \gamma_0 \right),
\end{align}
where $\{Z_i\}_{i=1}^{N}$ are independent Gaussian RVs with zero mean and unit variance. At a higher level of abstraction, this amounts to determining the CDF of a linear combination of non-central chi-squared RVs with degree 1. In the remainder of this paper, we consider the above expression of $P$. The following proposition gives a lower bound on $P$.
\begin{remark}
If the matrix $\Sigma$ is positive semi-definite, then $A$ is also positive semi-definite. Without loss of generality, we assume that the non-zero eigenvalues are $\{\lambda_i\}_{i=1}^{d}$, where $d < N$. In this case, the probability $P$ is
\begin{align}
P = \mathbb{P}\left(S_d = \sum_{i=1}^{d}{\lambda_i (Z_i + \alpha_i)^2} \leq \gamma_0 \right).
\end{align}
In the rest of this paper, we assume that $d=N$, i.e., the positive definite case, but the same reasoning applies when we simply replace $N$ with $d$ in the positive semi-definite case.
\end{remark}
\begin{proposition}\label{prop1}
Let $X = (X_1,\dots, X_N)^{T}$ be a Gaussian random vector with mean $\mu$ and covariance matrix $\Sigma_X$, and let $\Sigma \in \mathbb{R}^{N\times N}$ be a given matrix. For a fixed threshold $\gamma_0 > 0$, we have
\begin{align}
P = \mathbb{P}(X^{T} \Sigma X \leq \gamma_0) \geq \prod_{i=1}^{N}{\left[1-Q_{\frac{1}{2}}\left(\alpha_i,\sqrt{\frac{\gamma_0}{N \lambda_i}}\right)\right]},
\end{align}
where $Q_{\nu}(\cdot,\cdot)$ is the generalized Marcum-Q function defined as \cite[Eq.(2)]{TSAY}
\begin{align}
Q_{\mu}(a,b) = \frac{1}{a^{\mu-1}} \int_{b}^{\infty}{x^{\mu} \exp\left(-\frac{x^2+a^2}{2}\right) I_{\mu-1}(a x) dx}.
\end{align}
\end{proposition}
\begin{proof}
In this proof, we consider the expression of $P$ 
\begin{align}
P = \mathbb{P}\left(S_N = \sum\limits_{i=1}^{N}{\lambda_i (Z_i + \alpha_i)^2} \leq \gamma_0 \right).
\end{align} 
We have
\begin{align}
\bigcap\limits_{i=1}^{N} \left\lbrace (Z_i + \alpha_i)^2 \leq \frac{\gamma_0}{N \lambda_i} \right\rbrace \subset \left\lbrace \sum_{i=1}^{N}{\lambda_i (Z_i + \alpha_i)^2} \leq \gamma_0 \right\rbrace.
\end{align}
Using the independence of $\{Z_i\}_{i=1}^{N}$ and, thus, the independence of $\{(Z_i + \alpha_i)^2\}_{i=1}^{N}$, we can write
\begin{align}
P \geq \prod_{i=1}^{N}{\mathbb{P}\left((Z_i + \alpha_i)^2 \leq \frac{\gamma_0}{N \lambda_i} \right)} = \prod_{i=1}^{N}{F_{W_i}\left( \frac{\gamma_0}{N \lambda_i}\right)} ,
\end{align}
where $F_{W_i}(\cdot)$ is the CDF of the RV $W_i = (Z_i + \alpha_i)^2$, $\forall i=1,\dots,N$. This corresponds to the CDF of a non-central Chi-squared RV with 1 degree of freedom. Therefore, we can write
\begin{align}
P \geq \prod_{i=1}^{N}{\left[1-Q_{\frac{1}{2}}\left(\alpha_i,\sqrt{\frac{\gamma_0}{N \lambda_i}}\right)\right]}.
\end{align}
\end{proof}
\section{Review of IS}\label{section3}
Let $f_{Z}(\cdot)$ denote the PDF of $Z$; then, we can write $P = \mathbb{E}[\mathbbm{1}_{(S_N \leq \gamma_0)}]$, where $\mathbbm{1}_{(\cdot)}$ is the indicator function and $\mathbb{E}[\cdot]$ is the expectation w.r.t. the probability measure under which the PDF of $Z$ is $f_Z(\cdot)$. Therefore, the naive MC estimator of $P$ is given by
\begin{align}
\hat{P}_{MC} = \frac{1}{M} \sum_{i=1}^{M}{\mathbbm{1}_{(S_N(\omega_i) \leq \gamma_0)}},
\end{align}
where $M$ is the number of MC samples, and $\{S_N(\omega_i)\}_{i=1}^{N}$ are i.i.d. realizations of the RV $S_N$. For each realization of $S_N$, the sequence $\{Z_i(\omega_i)\}_{i=1}^{N}$ is sampled independently according to the PDF $f_Z(\cdot)$.

When dealing with rare event simulations, reducing the variance of the estimator of the quantity of interest causes the number of simulation runs required to achieve a certain accuracy level to become smaller. The IS method is a variance reduction technique that can be used to evaluate the probability of rare events. The core of the IS method is to propose the correct new PDF so that the new estimator has a smaller variance. More specifically, we can re-write $P$ as
\begin{align}
P = \mathbb{E}^*[\mathbbm{1}_{(S_N \leq \gamma_0)} \mathcal{L}(Z)],
\end{align}
where $\mathbb{E}^*[\cdot]$ is the expectation w.r.t. the probability measure under which the PDF of $Z$ is the biased PDF $f_Z^*(\cdot)$, and $\mathcal{L}(\cdot)$ is the likelihood ratio defined as
\begin{align}
\mathcal{L}(Z) = \frac{f_Z(Z)}{f_Z^*(Z)}.
\end{align}
Since $\{Z_i\}_{i=1}^{N}$ are independent, we can write the likelihood ratio in terms of the marginal PDFs of $\{Z_i\}_{i=1}^{N}$, i.e., 
\begin{align}
\mathcal{L}(Z_1,\dots,Z_N) = \prod_{i=1}^{N}{\frac{f_{Z_i}(Z_i)}{f_{Z_i}^*(Z_i)}}.
\end{align}
Thus, the IS estimator of $P$ is given by
\begin{align}\label{ISes}
\hat{P}_{IS} = \frac{1}{M^*} \sum_{i=1}^{M^*}{\mathbbm{1}_{(S_N(\omega_i) \leq \gamma_0)} \mathcal{L}(Z_1(\omega_i),\dots,Z_N(\omega_i))},
\end{align}
where the sequence $\{Z_i\}_{i=1}^{N}$ is sampled according to the biased PDF $f_Z^*(\cdot)$ for each realization $i=1,\dots,M^*$. The likelihood term can be interpreted as a weight that corrects the bias caused by the sampling from the biased PDF $f_Z^*(\cdot)$. In fact, it see that the IS estimator of $P$ is unbiased. 

The efficiency of the proposed IS estimator compared to naive MC can be measured by many criteria. When it comes to evaluating very low probabilities, IS estimators endowed with the bounded relative error property are desirable. A naive MC estimator requires a number of samples $M$ that grows as $\mathcal{O}(P^{-1})$. To achieve the same accuracy, the number of simulation runs $M^{*}$ needed by an IS estimator with a bounded relative error remains bounded, independently of $P$. Mathematically speaking, we say that the IS estimator satisfies the bounded relative error criterion if the following statement holds
\begin{align}
\underset{\gamma_0 \rightarrow 0}{\limsup} \frac{\mathbb{E}^*[\mathbbm{1}_{(S_N \leq \gamma_0)} \mathcal{L}^2(Z_1,\dots,Z_N)]}{P^2} < +\infty.
\end{align}
To have a clear idea of the gain that the proposed IS estimator achieves compared to the naive MC one, we determine the number of simulation runs required by both estimators when the accuracy requirement is fixed, e.g., when the relative error of both estimators is assumed to be the same. We start by defining the relative error of both estimators as
\begin{align}
\varepsilon &= \frac{C}{P} \sqrt{\frac{P(1-P)}{M}}, \\
\varepsilon^* &= \frac{C}{P} \sqrt{\frac{\mathbb{V}^*[\mathbbm{1}_{(S_N \leq \gamma_0)} \mathcal{L}(Z_1,\dots,Z_N)]}{M^*}},
\end{align}
where we take $C = 1.96$, which corresponds to a $95\%$ confidence interval, and $\mathbb{V}^*$ denotes the variance w.r.t. the probability measure under which the PDF of Z is $f_Z^*(\cdot)$.
\section{Proposed IS Scheme}\label{section4}
\subsection{Real Valued Case}
Our approach consists of shifting the mean and scaling the variance of each variate $\{Z_i\}_{i=1}^{N}$ so that the marginal biased PDF is written as
\begin{align}\label{bias}
f_{Z_i}^*(z) = \frac{1}{\sqrt{2\pi} \sigma_i} \exp\left(-\frac{1}{2} \left(\frac{z+\alpha_i}{\sigma_i}\right)^2\right).
\end{align}
While the original PDF of $Z_i$, $\forall i=1,\dots, N$, is a standard Gaussian, the biased PDF corresponds to a Gaussian with mean $-\alpha_i$ and variance $\sigma_i^2$. In our approach, we choose the parameter $\sigma_i$, hoping that the event of interest becomes no longer rare. A possible solution is to look for $\sigma_i$ in the form  $\sigma_i^2 = \theta \frac{\gamma_0}{\lambda_i}$, where $\theta$ is a positive parameter such that the mean of $\sum\limits_{i=1}^{N}{\lambda_i (Z_i+ \alpha_i)^2}$ under the biased PDF is equal to $\gamma_0$. That is,
\begin{align}
\mathbb{E}^*\left[\sum\limits_{i=1}^{N}{\lambda_i (Z_i+ \alpha_i)^2} \right] = \gamma_0.
\end{align}
Using the linearity of the expected value and the fact that, under the new probability measure, $\{Z_i + \alpha_i\}_{i=1}^{N}$ are zero mean Gaussian RVs with variance $\sigma_i^2$, we get
\begin{align}\label{choice}
\sigma_i = \sqrt{\frac{\gamma_0}{N \lambda_i}},~ i=1,\dots, N.
\end{align} 
The above value of $\sigma_i$ is clearly non-negative since the eigenvalues $\{\lambda_i\}_{i=1}^{N}$ are all non-negative. As the threshold $\gamma_0$ approaches zero, the values of $\sigma_i$ become smaller, leading to the reduction of the variance of the IS estimator. Defining the biased PDFs using the values of $\sigma_i$ obtained in (\ref{choice}), we show that our proposed IS estimator satisfies the bounded relative error property.
\begin{proposition}
Let the marginal biased PDFs be defined as in (\ref{bias}), and $\sigma_i$ as in (\ref{choice}). Then, the IS estimator (\ref{ISes}) of the probability $P$, given by (\ref{P}), satisfies
\begin{align}
\underset{\gamma_0 \rightarrow 0}{\limsup} \frac{\mathbb{E}^*[\mathbbm{1}_{(S_N \leq \gamma_0)} \mathcal{L}^2(Z_1,\dots,Z_N)]}{P^2} \leq \prod_{i=1}^{N}{\frac{\pi e}{\alpha_i^{2}}} < +\infty.
\end{align}
\end{proposition}
\begin{proof}
We recall the definition of the likelihood ratio,
\begin{align}
\nonumber &\mathcal{L}(Z_1,\dots,Z_N) = \prod_{i=1}^{N}{\frac{f_{Z_i}(Z_i)}{f_{Z_i}^*(Z_i)}} \\
&= \left(\prod_{i=1}^{N}{\sigma_i}\right) \exp\left(\frac{1}{2} \sum_{i=1}^{N}{\left(\frac{Z_i+\alpha_i}{\sigma_i}\right)^2}-\frac{1}{2} \sum_{i=1}^{N}{Z_i^2}\right).
\end{align}
A trivial upper bound for the likelihood ratio is given by
\begin{align}
\mathcal{L}(Z_1,\dots,Z_N) \leq \left(\prod_{i=1}^{N}{\sigma_i}\right) \exp\left(\frac{1}{2} \sum_{i=1}^{N}{\left(\frac{Z_i+\alpha_i}{\sigma_i}\right)^2}\right).
\end{align}
With the choice of $\sigma_i$ given in (\ref{choice}), we get
\begin{align}
\nonumber \mathcal{L}(Z_1,\dots,Z_N) &\leq \left(\frac{\gamma_0}{N}\right)^{\frac{N}{2}} \left(\prod_{i=1}^{N}{\frac{1}{\sqrt{\lambda_i}}}\right) \times \\
&\exp\left(\frac{N}{2 \gamma_0} \sum_{i=1}^{N}{\lambda_i \left(Z_i+\alpha_i\right)^2}\right).
\end{align}
Using the above upper bound of the likelihood ratio, we write
\begin{align}
\nonumber &\mathbbm{1}_{\left(\sum\limits_{i=1}^{N}{\lambda_i (Z_i+\alpha_i)^2} \leq \gamma_0\right)} \mathcal{L}(Z_1,\dots,Z_N) \\
&\leq \left(\frac{\gamma_0}{N}\right)^{\frac{N}{2}} \left(\prod_{i=1}^{N}{\frac{1}{\sqrt{\lambda_i}}}\right) e^{\frac{N}{2}}.
\end{align}
Thus, we obtain the upper bound
\begin{align}\label{1}
\mathbb{E}^*[\mathbbm{1}_{\left(S_N \leq \gamma_0\right)} \mathcal{L}^2(Z_1,\dots,Z_N)] \leq \left(\frac{\gamma_0}{N}\right)^{N} \left(\prod_{i=1}^{N}{\frac{1}{\lambda_i}}\right) e^{N}.
\end{align}
From Proposition \ref{prop1}, we have 
\begin{align}
P \geq  \prod_{i=1}^{N}{\left[1-Q_{\frac{1}{2}}\left(\alpha_i,\sqrt{\frac{\gamma_0}{N \lambda_i}}\right)\right]}.
\end{align}
Using \citep[Eq.(8)]{TSAY}, we have the  asymptotic expansion around $b=0$ of $Q_{\nu}(a,b)$, i.e.,
\begin{align}
Q_{\nu}(a,b) \underset{b \rightarrow 0}{\sim} 1-\frac{1}{\Gamma(\nu+1)} \left(\frac{b^2}{2}\right)^{\nu} \left(\frac{a^2}{2}\right)^{1-\nu}.
\end{align}
Therefore, as $\gamma_0 \rightarrow 0$, we have
\begin{align}
P \geq  \left(\frac{1}{N \pi}\right)^{\frac{N}{2}} \gamma_0^{\frac{N}{2}} \left(\prod_{i=1}^{N}{\frac{\alpha_i}{\sqrt{\lambda_i}}}\right),
\end{align}
and we can write
\begin{align}\label{2}
\frac{1}{P^2} \leq \left(N \pi\right)^N \gamma_0^{-N} \left(\prod_{i=1}^{N}{\frac{\lambda_i}{\alpha_i^{2}}}\right).
\end{align}
By combining (\ref{1}) and (\ref{2}), we obtain
\begin{align}
\nonumber &\underset{\gamma_0 \rightarrow 0} \limsup \frac{\mathbb{E}^*[\mathbbm{1}_{\left(S_N \leq \gamma_0\right)} \mathcal{L}^2(Z_1,\dots,Z_N)]}{P^2} \\
&\leq \pi^N \left(\prod_{i=1}^{N}{\frac{1}{\alpha_i^{2}}}\right) e^{N} < +\infty.
\end{align}
\end{proof}
\begin{remark}
If $\mu$ is zero, then $\alpha = Q \Sigma^{-\frac{1}{2}} \mu$ is also zero. In this case, we use the same IS scheme, i.e., we introduce the biased PDF to be Gaussian with zero mean and variance $\sigma_i^2$. In the proof, we use the following lower bound, derived using similar reasoning but involving the CDFs of central Chi-squared RVs
\begin{align}
P = \mathbb{P}(X^{T} \Sigma X \leq \gamma_0) \geq \frac{1}{\pi^{\frac{N}{2}}} \prod_{i=1}^{N}{\gamma\left(\frac{1}{2}, \frac{\gamma_0}{ 2 N \lambda_i}\right)},
\end{align}
where $\gamma(\cdot,\cdot)$ is the lower incomplete Gamma function defined as \cite[Eq. (8.350.1)]{TII}
\begin{align}
\gamma(a, x) = \int_{0}^{x}{e^{-t} t^{a-1} dt}.
\end{align}
As $\gamma_0 \rightarrow 0$, we have the upper bound
\begin{align}
\frac{1}{P^2} \leq \left(\frac{N \pi}{2}\right)^N \left(\prod_{i=1}^{N}{\lambda_i}\right) \gamma_0^{-N}.
\end{align}
The proposed IS estimator also has the bounded relative error property in this case, since
\begin{align}
\underset{\gamma_0 \rightarrow 0} \limsup \frac{\mathbb{E}^*[\mathbbm{1}_{\left(S_N \leq \gamma_0\right)} \mathcal{L}^2(Z_1,\dots,Z_N)]}{P^2} \leq \left(\frac{\pi}{2}\right)^N e^{N} < +\infty.
\end{align}
\end{remark}
\subsection{Complex Valued Case}
In this section, we briefly show how the proposed approach is still valid even if we consider the complex case for which the probability is
\begin{align}
P = \mathbb{P}(X^{*} \Sigma X \leq \gamma_0),
\end{align}
where now $X$ is a complex Gaussian random vector, $X^{*}$ is its conjugate transpose, and $\Sigma$ is a Hermitian positive definite matrix. The complex setting is of paramount importance in many applications involving wireless techniques \cite{Turin1960}, \cite{Slichenko2014}, \cite{Al-Naffouri2016}, \cite{Biyari1993}, \cite{Pablo2018}. Using similar arguments, we write
\begin{align}
P = \mathbb{P}\left(S_N = \sum_{i=1}^{N}{\lambda_i |Z_i + \alpha_i|^2} \leq \gamma_0 \right),
\end{align}
where $|.|$ is the module of a complex number. Using the fact that the random variable $|Z_i + \alpha_i|^2$ has a central Chi-squared distribution with 2 degrees of freedom, we obtain the bound
\begin{align}
\frac{1}{P^2} \leq (2 N)^{2N} \left(\prod_{i=1}^{N}{\lambda_i^2}\right) \gamma_0^{-2N}.
\end{align}
We write both the original and biased PDF of $Z_i$, $i=1,\dots,N$, in the complex scenario as
\begin{align}
f_{Z_i}(z) = \frac{1}{\pi} e^{-|z|^2}, ~ f^{*}_{Z_i}(z) & = \frac{1}{\pi \sigma_i^2} e^{-\frac{|z+\alpha_i|^2}{\sigma_i^2}}.
\end{align}
Using a similar manipulation to the real-valued case, and using the same expression of $\sigma$ as in (\ref{choice}), we get
\begin{align}
\mathbb{E}^*[\mathbbm{1}_{\left(S_N \leq \gamma_0\right)} \mathcal{L}^2(Z_1,\dots,Z_N)] \leq \left(\frac{\gamma_0}{N}\right)^{2N} e^{2N} \left(\prod_{i=1}^{N}{\frac{1}{\lambda_i^2}}\right).
\end{align}
Thus, we can say that the proposed IS estimator maintains the bounded relative error in the complex-valued case since
\begin{align}
\underset{\gamma_0 \rightarrow 0} \limsup \frac{\mathbb{E}^*[\mathbbm{1}_{\left(S_N \leq \gamma_0\right)} \mathcal{L}^2(Z_1,\dots,Z_N)]}{P^2} \leq 2^{2N} e^{2N} < +\infty.
\end{align}
\begin{remark}
The upper bound for the relative error of both the real-valued and complex-valued cases suffers from an exponential deterioration w.r.t. $N$, i.e. the size of the Gaussian RV $X$.
\end{remark}
\section{Numerical Examples}\label{section5}
To show the accuracy and efficiency of the proposed IS scheme compared to naive MC, we consider three examples. The first example is a toy example where we compute the probability $P$ for a given scenario. In this example, we consider real Gaussian random vectors and the non-zero mean case. The second example is inspired by the wireless communication field. We estimate the outage probability of diversity receivers over correlated Gamma fading channels. This case corresponds to the real case, but with zero-mean Gaussian random vectors. In the third example, we show how the IS approach can be extended to the complex case by estimating the outage probability of diversity receivers over correlated Rician fading channels.
\subsection{Example 1: Toy Example}
\subsubsection{Problem Setup}
In this example, we compute the probability $P$ when the matrix $\Sigma$ is defined as
\[
\Sigma=
  \begin{pmatrix}
    \phantom{-}17 & -5 & \phantom{-}9  \\
    -5 & \phantom{-}18 & -5  \\
    \phantom{-}9 & -5 & \phantom{-}18 
  \end{pmatrix}
\].
The mean and covariance matrix of the $(3 \times 3)$ random vector $X$ are given by
\[
\mu =
\begin{pmatrix}
    1   \\
    2   \\
    0  
  \end{pmatrix}
,
\hspace{0.3cm}
\Sigma_X=
  \begin{pmatrix}
    2 & 1 & 0  \\
    1 & 2 & 1  \\
    0 & 1 & 2 
  \end{pmatrix}
\].
To have evaluate the efficiency of the proposed IS scheme, we introduce a metric that measures the improvement in terms of the number of simulation runs. The efficiency indicator, $\xi$, is defined as
\begin{align}
\xi = \frac{\mathbb{V}[\hat{P}_{MC}]}{\mathbb{V^*}[\hat{P}_{IS}]} = \frac{M}{M^*}.
\end{align}
For a fixed number of simulation runs $M = M^*$, this metric can also be interpreted as a measure of variance reduction. 
\subsubsection{Results and Discussions}
In Table \ref{tab:table1}, we provide a comparison between the efficiencies of the MC and IS estimators of $P$. For a specific range of the threshold $\gamma_0$, we compute the MC and IS estimates, their relative errors, and the efficiency indicator, for the same number of simulation runs $M = M^* = 5 \times 10^6$.
\begin{table*}[!htp]
\small
\center
\caption {\label{tab:table1} Comparison of the efficiency of MC and IS estimators of $P$ as a function of the threshold $\gamma_0$ using $M = M^* = 5 \times 10^6$} 
\begin{tabular}{ccccccccc}
\hline
\hline
  & \multicolumn{8}{c}{$\gamma_0$}    \\ 
\cline{2-9} 
        & $10^{-2}$ & $2.3 \times 10^{-2}$   & $3.6  \times 10^{-2}$  &  $4.9 \times 10^{-2}$ &  $6.1 \times 10^{-2}$   & $7.4 \times 10^{-2}$  & $8.7 \times 10^{-2}$ & $10^{-1}$                               \\    \hline
$\hat{P}_{MC} (\times 10^{-5})$ & $0.02$ &  $ 0.2 $   & $ 0.34 $  &  $ 0.58 $  &  $ 0.8 $   & $ 1.14 $   & $ 1.28 $ & $ 1.62 $ \\
$\hat{P}_{IS} (\times 10^{-5})$ & $ 0.06$ &  $0.19$   & $0.38 $  &  $ 0.59 $  &  $ 0.85 $   & $ 1.13 $   & $ 1.44 $ & $ 1.76 $ \\
$\varepsilon (\times 10^{-7})$ & $6.55$  &  $12.18$ &   $17.02$  &  $21.43$   & $25.55$ &  $29.47$  &  $33.22$ &   $36.82$ \\
$\varepsilon^* (\times 10^{-8})$ & $0.05$  &  $0.16 $  & $0.32$ &   $0.50$  &  $0.71$  & $0.95$  & $ 1.21 $ &  $1.48 $\\
$\xi (\times 10^{5})$ & $7.01$ &   $5.86$ &  $2.61$ &   $1.77$ &   $1.21$  & $ 0.974$   & $0.68$  &  $0.57$\\
 \hline\hline
\end{tabular}
\end{table*}

For relatively high values of $P$, the IS and MC estimates match. However, as the threshold becomes smaller, i.e., the probability of interest becomes smaller, the MC estimate becomes less accurate unlike the proposed IS scheme. While both estimates are built using the same number of simulation runs, the IS relative error remains bounded no matter how small the probability $P$ becomes. On the other hand, MC relative error tends to increase as the threshold decreases, which is consistent with the observation that the MC method becomes less accurate when estimating rare events. Finally, we quantify the gain in terms of the number of simulation runs (or, equivalently, in terms of variance reduction) by considering the efficiency indicator. We observe that the smaller the probability becomes, the larger the gain (or, equivalently, the smaller the reduction) becomes, reflecting the efficiency of the proposed scheme compared to naive MC. In fact, when the probability of interest is of the order of $1.7 \times 10^{-5}$, a gain in terms of the number of simulation runs of the order of $5.6 \times 10^4$ is achieved. This gain tends to increase as the probability becomes smaller since naive MC requires more samples to estimate the probability than the proposed IS estimator. In contrast, the bounded relative error property of the IS estimator causes the number of required simulation runs to remain almost constant no matter how small the probability becomes.
\subsection{Example 2: Maximum Ratio Combining Over Correlated Nakagami-m Fading Channel}
\subsubsection{Problem Setup}
In general, the problem of finding the left tail of quadrature form in Gaussian random vectors has many applications in the wireless communication filed \cite{Holm2004}, \cite{Simon2001}. To combat the attenuation of the received signal in wireless communication systems, different diversity techniques can be used. In fact, there are in particular more or less complex linear combination techniques which make it possible to recover a signal with a good average level. Among these diversity techniques, we find the maximum ratio combining (MRC) technique. The MRC technique improves the average power of the output signal by forming it from the maximum signal of all the branches. The instantaneous signal-to-noise ratio (SNR) expression at the MRC diversity receiver, is given by
\begin{align}
\gamma_{end} = \frac{E_s}{N_0} \sum_{\ell=1}^{L}{R_{\ell}},
\end{align}
where $\frac{E_s}{N_0}$ is the average SNR at each branch and where, for each $\ell=1,\dots,L$, $R_{\ell}$ follows a Gamma distribution with PDF
\begin{align}
f_{R_\ell}(r) = \frac{r^{k-1}}{\theta^k \Gamma(k)} \exp\left(-\frac{r}{\theta}\right),
\end{align}
where $k, \theta > 0$ are respectively the shape and scale parameters, respectively, of the PDF, and $\Gamma(\cdot)$ is the Gamma function defined as \cite[Eq. (8.310.1)]{TII}
\begin{align}
\Gamma(a) = \int_{0}^{\infty}{e^{-t} t^{a-1} dt}, ~ a > 0.
\end{align}
Assuming that the shape parameter $k$ is a multiple of 0.5, then we can write
\begin{align}
R_{\ell} = \sum_{m=1}^{2 k}{X_{\ell,m}^2},
\end{align}
where $X_{\ell,m}$, $\forall m=1,\dots,2k$ are independent zero-mean Gaussian RVs. 

To quantify the quality of a communication system, we compute a metric called the outage probability. Depending on the transmission technique used and the channel over which the signal is transmitted, this metric measures the probability that the instantaneous SNR drops below a certain threshold $\gamma_{th}$, i.e.,
\begin{align}
P_{out} = \mathbb{P}(\gamma_{end} \leq \gamma_{th}).
\end{align}
In the MRC scenario, the outage probability can be written as
\begin{align}
P_{out} = \mathbb{P}\left(\sum_{\ell=1}^{L}{\sum_{m=1}^{2 k}{X_{\ell,m}^2}} \leq \gamma_0 \right),
\end{align}
where $\gamma_0 = \frac{N_0}{E_s} \gamma_{th}$.
To facilitate the modeling of the channel correlation, we introduce the ($2kL \times 1$) vector $X=[X_{1,1},X_{2,1},\dots,X_{L,2k}]^T$. The joint pdf of the Gaussian vector $X$  is
\begin{align}
f_X(X) = \frac{1}{\sqrt{(2\pi)^{2kL} |\Sigma_X|}} \exp\left(-\frac{1}{2} X^{T} \Sigma_X^{-1} X\right). 
\end{align}
Thus, we re-write the outage probability as in (\ref{P}) where $\Sigma$ is the identity matrix of order $2kL \times 2kL$ and $N=2kL$. In other words, the outage probability is given by
\begin{align}
P_{out} = \mathbb{P}\left(X^T X \leq \gamma_0 \right).
\end{align}
\subsubsection{Results and Discussions}
In this section, we discuss the efficiency of the proposed IS estimator for the purpose of computing the outage probability of MRC diversity receivers over correlated Gamma fading channels, and we consider the parameter $k=1.5$. In Fig. \ref{fig21}, we plot the outage probability of $L$-branch MRC diversity receivers over the correlated Gamma fading model as a function of the threshold $\gamma_{th}$ for different numbers of branches $L \in \{2, 3, 4\}$. The number of simulation runs required to construct the naive MC estimator and the proposed IS estimator are $M = 10^7$ and $M^* = 10^4$, respectively. While naive MC presents an accurate estimate for relatively high values of the outage probability, it tends to become erroneous as the outage probability becomes smaller. Although the proposed IS scheme is constructed using fewer of simulation runs, i.e. $M^* = 10^4$ compared to $M = 10^7$, it provides an accurate estimate for the outage probabilities even when the probability is very small. We also observe that the outage probability becomes smaller as $L$ increases, a well-known observation when using diversity techniques.

To compare the efficiency of both methods, we investigate in Fig. \ref{fig22} the number of simulation runs required to achieve a $5\%$ relative error for $L$-branch diversity receivers over the correlated Gamma fading model for the three different numbers of branches. For this fixed accuracy requirement, we can see that the number of required simulation runs by naive MC tends to increase rapidly while it remains bounded for the proposed IS scheme. In fact, for $\gamma_{th} = 5 $ dB and $L=2$ (which corresponds to  a probability of the order of $10^{-3}$), we already observe a reduction of the order of $10^3$ simulation runs for an accuracy requirement is set to a level of $5\%$. In Fig. \ref{fig23}, we plot for a fixed number of branches $L=3$, the plot of the outage probability estimates along with its error bars. As in the previous example, the error bars seem to have the same magnitude for the proposed IS scheme, and they are relatively small. The error bars of the naive MC tend to increase as the probability becomes smaller, indicating that the naive MC become less accurate as the events become more rare.
\begin{figure}[H]
\centering
\setlength\figureheight{0.3\textwidth}
\setlength\figurewidth{0.35\textwidth}
\scalefont{0.7}
\begin{tikzpicture}

\begin{axis}[%
width=\figurewidth,
height=\figureheight,
at={(0\figurewidth,0\figureheight)},
scale only axis,
every outer x axis line/.append style={darkgray!60!black},
every x tick label/.append style={font=\color{darkgray!60!black}},
xmin=-5,
xmax=10,
xlabel={$\gamma{}_{\text{th}}\text{(dB)}$},
xmajorgrids,
every outer y axis line/.append style={darkgray!60!black},
every y tick label/.append style={font=\color{darkgray!60!black}},
ymode=log,
ymin=1e-14,
ymax=1,
yminorticks=true,
ymajorgrids,
yminorgrids,
grid style={dotted},
ylabel={Outage Probability},
legend style={at={(0.018574218750001,0.824593784825617)},anchor=south west,draw=black,fill=white,align=left}
]
\addplot [
color=blue,
solid,
line width=1.0pt,
mark size=1.5pt,
mark=o,
mark options={solid},
]
  table[row sep=crcr]{%
-5	5e-07\\
-4	1e-06\\
-3	2e-06\\
-2	4.1e-06\\
-1	8.1e-06\\
0	2.02e-05\\
1	4.26e-05\\
2	8.32e-05\\
3	0.0001624\\
4	0.0003135\\
5	0.0006137\\
6	0.0011804\\
7	0.0022506\\
8	0.0042778\\
9	0.0079966\\
10	0.014798\\
};
\addlegendentry{Naive MC};
\addplot [
color=red,
solid,
line width=1.0pt,
mark size=1.5pt,
mark=asterisk,
mark options={solid},
]
  table[row sep=crcr]{%
-5	6.67388767460473e-07\\
-4	1.32744118471057e-06\\
-3	2.63814742690779e-06\\
-2	5.2376825842133e-06\\
-1	1.03853508952747e-05\\
0	2.05589634527906e-05\\
1	4.06161352124831e-05\\
2	8.00361991535461e-05\\
3	0.000157210004486492\\
4	0.000307555173198044\\
5	0.000598643735402039\\
6	0.00115786573072826\\
7	0.00222175422015217\\
8	0.00422100254043595\\
9	0.00792032554194431\\
10	0.0146336622777613\\
};
\addlegendentry{Proposed IS};
\addplot [
color=blue,
solid,
line width=1.0pt,
mark size=1.5pt,
mark=o,
mark options={solid},
forget plot
]
  table[row sep=crcr]{%
-5	0\\
-4	0\\
-3	0\\
-2	0\\
-1	0\\
0	0\\
1	0\\
2	5e-07\\
3	9e-07\\
4	1.9e-06\\
5	5.9e-06\\
6	1.15e-05\\
7	3.32e-05\\
8	8.97e-05\\
9	0.0002337\\
10	0.0005997\\
};
\addplot [
color=red,
solid,
line width=1.0pt,
mark size=1.5pt,
mark=asterisk,
mark options={solid},
forget plot
]
  table[row sep=crcr]{%
-5	1.64309365186566e-10\\
-4	4.61482326856602e-10\\
-3	1.29496451210476e-09\\
-2	3.62969307146117e-09\\
-1	1.01593141605601e-08\\
0	2.83845410229236e-08\\
1	7.9126485945568e-08\\
2	2.19954283115376e-07\\
3	6.09251982118009e-07\\
4	1.68003367925761e-06\\
5	4.60677307810558e-06\\
6	1.25431979276872e-05\\
7	3.38508504299658e-05\\
8	9.03445897836258e-05\\
9	0.000237785306155984\\
10	0.00061503361800638\\
};
\addplot [
color=blue,
solid,
line width=1.0pt,
mark size=1.5pt,
mark=o,
mark options={solid},
forget plot
]
  table[row sep=crcr]{%
-5	0\\
-4	0\\
-3	0\\
-2	0\\
-1	0\\
0	0\\
1	0\\
2	0\\
3	0\\
4	0\\
5	0\\
6	1e-07\\
7	4e-07\\
8	1.3e-06\\
9	4.8e-06\\
10	1.57e-05\\
};
\addplot [
color=red,
solid,
line width=1.0pt,
mark size=1.5pt,
mark=asterisk,
mark options={solid},
forget plot
]
  table[row sep=crcr]{%
-5	2.41686281859341e-14\\
-4	9.58689691990772e-14\\
-3	3.79924043226093e-13\\
-2	1.50384407200875e-12\\
-1	5.94379170666562e-12\\
0	2.34483515397302e-11\\
1	9.22867213122575e-11\\
2	3.62143286089067e-10\\
3	1.41580763674747e-09\\
4	5.50926573548693e-09\\
5	2.13120417389961e-08\\
6	8.18352855216979e-08\\
7	3.11325871126616e-07\\
8	1.17061637789828e-06\\
9	4.3375440142031e-06\\
10	1.57791885408414e-05\\
};
\node[black,draw,below] at (axis cs:-3.3,5.0e-05){{$L = 2$}};
\node[black,draw,below] at (axis cs:1,8.0e-06){{$L = 3$}};
\node[black,draw,below] at (axis cs:7,5.0e-08){{$L = 4$}};
\draw (axis cs:-4,1.0e-06) ellipse (0.1cm and 0.2cm);
\draw (axis cs:3,1.0e-06) ellipse (0.1cm and 0.2cm);
\draw (axis cs:7,5.0e-07) ellipse (0.1cm and 0.2cm);
\end{axis}
\end{tikzpicture}%
\caption{Outage probability of $L$-branch MRC diversity receivers over the correlated Nakagami-m fading model with $E_s/N_0 = 10$ dB. Number of samples $M = 10^7$ and $M^* = 10^4$.}
\label{fig21}
\end{figure}
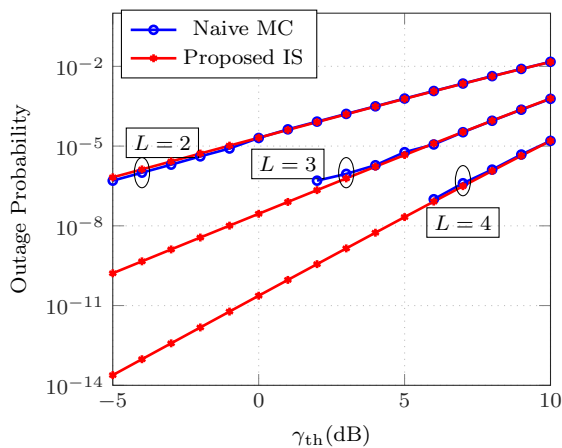
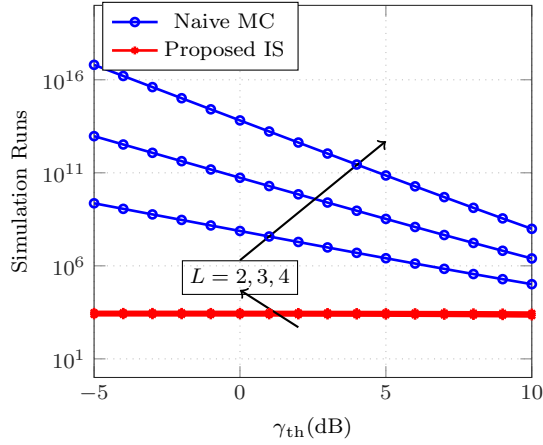
\begin{figure}[H]
\centering
\setlength\figureheight{0.3\textwidth}
\setlength\figurewidth{0.35\textwidth}
\scalefont{0.7}
\begin{tikzpicture}

\begin{axis}[%
width=\figurewidth,
height=\figureheight,
at={(0\figurewidth,0\figureheight)},
scale only axis,
every outer x axis line/.append style={darkgray!60!black},
every x tick label/.append style={font=\color{darkgray!60!black}},
xmin=-5,
xmax=10,
xlabel={$\gamma{}_{\text{th}}\text{(dB)}$},
xmajorgrids,
every outer y axis line/.append style={darkgray!60!black},
every y tick label/.append style={font=\color{darkgray!60!black}},
ymode=log,
ymin=1,
ymax=1e+20,
yminorticks=true,
ymajorgrids,
yminorgrids,
grid style={dotted},
ylabel={Simulation Runs},
legend style={at={(0.018574218750001,0.824593784825617)},anchor=south west,draw=black,fill=white,align=left}
]
\addplot [
color=blue,
solid,
line width=1.0pt,
mark size=1.5pt,
mark=o,
mark options={solid},
]
  table[row sep=crcr]{%
-5	2302464544.48266\\
-4	1157594007.10158\\
-3	582467806.933084\\
-2	293380121.238911\\
-1	147960724.384725\\
0	74741531.1964974\\
1	37831703.5725676\\
2	19197775.874229\\
3	9772904.9104551\\
4	4994770.14951548\\
5	2565332.28241179\\
6	1325594.78743546\\
7	690097.018692794\\
8	362509.575390643\\
9	192475.587384126\\
10	103470.566728778\\
};
\addlegendentry{Naive MC};
\addplot [
color=red,
solid,
line width=1.0pt,
mark size=1.5pt,
mark=asterisk,
mark options={solid},
]
  table[row sep=crcr]{%
-5	2205.99177957114\\
-4	2203.26642054436\\
-3	2199.8401228364\\
-2	2195.53416223362\\
-1	2190.12515888733\\
0	2183.33446519261\\
1	2174.81536431794\\
2	2164.13784178662\\
3	2150.77082641783\\
4	2134.06208295503\\
5	2113.21648848399\\
6	2087.27442520976\\
7	2055.09378750555\\
8	2015.34216338343\\
9	1966.51102256478\\
10	1906.97282156547\\
};
\addlegendentry{Proposed IS};
\addplot [
color=blue,
solid,
line width=1.0pt,
mark size=1.5pt,
mark=o,
mark options={solid},
forget plot
]
  table[row sep=crcr]{%
-5	9352114518869.48\\
-4	3329791651519.41\\
-3	1186627111126.42\\
-2	423352598737.468\\
-1	151254303204.27\\
0	54136508853.2583\\
1	19420044502.7743\\
2	6986177492.18101\\
3	2522173269.68235\\
4	914646793.909557\\
5	333559499.240652\\
6	122506296.598292\\
7	45392891.5791422\\
8	17007118.8166271\\
9	6460763.42740613\\
10	2496928.41135929\\
};
\addplot [
color=red,
solid,
line width=1.0pt,
mark size=1.5pt,
mark=asterisk,
mark options={solid},
forget plot
]
  table[row sep=crcr]{%
-5	2766.36264690286\\
-4	2763.78771072396\\
-3	2760.54911564524\\
-2	2756.47680997524\\
-1	2751.35775439545\\
0	2744.92539969222\\
1	2736.84679413649\\
2	2726.7069122929\\
3	2713.98981019635\\
4	2698.05629663531\\
5	2678.11802459224\\
6	2653.20834590915\\
7	2622.15109249284\\
8	2583.52990595167\\
9	2535.66326153313\\
10	2476.59464726647\\
};
\addplot [
color=blue,
solid,
line width=1.0pt,
mark size=1.5pt,
mark=o,
mark options={solid},
forget plot
]
  table[row sep=crcr]{%
-5	6.35799428986322e+16\\
-4	1.60285440934379e+16\\
-3	4.04459793318466e+15\\
-2	1.02180806414666e+15\\
-1	258528574994916\\
0	65532964966016.1\\
1	16650716137794.9\\
2	4243182349280.35\\
3	1085345182453.27\\
4	278919200000.874\\
5	72101959355.6525\\
6	18777228727.8353\\
7	4935791285.32901\\
8	1312674442.45306\\
9	354263456.399457\\
10	97382431.8716062\\
};
\addplot [
color=red,
solid,
line width=1.0pt,
mark size=1.5pt,
mark=asterisk,
mark options={solid},
forget plot
]
  table[row sep=crcr]{%
-5	3358.00859695837\\
-4	3355.58610779275\\
-3	3352.53830148766\\
-2	3348.70439396138\\
-1	3343.88262829971\\
0	3337.82005028872\\
1	3330.19986089029\\
2	3320.62584896628\\
3	3308.60335492539\\
4	3293.51619381899\\
5	3274.59900805483\\
6	3250.9046823337\\
7	3221.26683198085\\
8	3184.25812912543\\
9	3138.14662303114\\
10	3080.8546838041\\
};
\draw[thick,arrows={->}] (axis cs:0,2e06) -- (axis cs:5,5.0e12);
\draw[thick,arrows={->}] (axis cs:2,5e02) -- (axis cs:0,5.0e04);
\node[black,draw,below] at (axis cs:0,2.0e06){{$L=2,3,4$}};
\end{axis}
\end{tikzpicture}%
\caption{Number of required simulation runs for $5\%$ relative error for $L$-branch MRC diversity receivers over correlated Nakagami-m fading model with  $E_s/N_0 = 10$ dB.}
\label{fig22}
\end{figure}
\begin{figure}[H]
\centering
\setlength\figureheight{0.3\textwidth}
\setlength\figurewidth{0.35\textwidth}
\scalefont{0.7}
\begin{tikzpicture}

\begin{axis}[%
width=\figurewidth,
height=\figureheight,
at={(0\figurewidth,0\figureheight)},
scale only axis,
every outer x axis line/.append style={darkgray!60!black},
every x tick label/.append style={font=\color{darkgray!60!black}},
xmin=-6,
xmax=12,
xmajorgrids,
xlabel={$\gamma{}_{\text{th}}\text{(dB)}$},
xmajorgrids,
every outer y axis line/.append style={darkgray!60!black},
every y tick label/.append style={font=\color{darkgray!60!black}},
ymode=log,
ymin=1e-10,
ymax=0.001,
yminorticks=true,
ymajorgrids,
yminorgrids,
grid style={dotted},
ylabel={Outage Probability},
legend style={at={(0.018574218750001,0.824593784825617)},anchor=south west,draw=black,fill=white,align=left}
]
\addplot [
color=red,
solid,
line width=1.0pt,
]
 plot [error bars/.cd, y dir = both, y explicit]
 table[row sep=crcr, y error plus index=2, y error minus index=3]{%
-5	1.64309365186566e-10	4.3210259851127e-12	4.3210259851127e-12\\
-4	4.61482326856602e-10	1.21304641528047e-11	1.21304641528047e-11\\
-3	1.29496451210476e-09	3.40193185998779e-11	3.40193185998779e-11\\
-2	3.62969307146117e-09	9.52833632861424e-11	9.52833632861424e-11\\
-1	1.01593141605601e-08	2.66445216074701e-10	2.66445216074701e-10\\
0	2.83845410229236e-08	7.43561942610096e-10	7.43561942610096e-10\\
1	7.9126485945568e-08	2.06974635569272e-09	2.06974635569272e-09\\
2	2.19954283115376e-07	5.74277309698781e-09	5.74277309698781e-09\\
3	6.09251982118009e-07	1.58697856726639e-08	1.58697856726639e-08\\
4	1.68003367925761e-06	4.36328414527261e-08	4.36328414527261e-08\\
5	4.60677307810558e-06	1.1920148908618e-07	1.1920148908618e-07\\
6	1.25431979276872e-05	3.23045694389142e-07	3.23045694389142e-07\\
7	3.38508504299658e-05	8.66699297635327e-07	8.66699297635327e-07\\
8	9.03445897836258e-05	2.29603709976858e-06	2.29603709976858e-06\\
9	0.000237785306155984	5.98688350497137e-06	5.98688350497137e-06\\
10	0.00061503361800638	1.5303695802127e-05	1.5303695802127e-05\\
};
\addlegendentry{Proposed IS};
\addplot [
color=blue,
solid,
line width=1.0pt,
]
 plot [error bars/.cd, y dir = both, y explicit]
 table[row sep=crcr, y error plus index=2, y error minus index=3]{%
-5	0	7.94487795499088e-09	7.94487795499088e-09\\
-4	0	1.33147681393038e-08	1.33147681393038e-08\\
-3	0	2.23041154571517e-08	2.23041154571517e-08\\
-2	0	3.73414365721425e-08	3.73414365721425e-08\\
-1	0	6.24724106167752e-08	6.24724106167752e-08\\
0	0	1.04423201300551e-07	1.04423201300551e-07\\
1	0	1.74348009554521e-07	1.74348009554521e-07\\
2	5e-07	2.906847412851e-07	2.906847412851e-07\\
3	9e-07	4.83787245444608e-07	4.83787245444608e-07\\
4	1.9e-06	8.03368317726648e-07	8.03368317726648e-07\\
5	5.9e-06	1.33031191564382e-06	1.33031191564382e-06\\
6	1.15e-05	2.19511605053367e-06	2.19511605053367e-06\\
7	3.32e-05	3.60606468326431e-06	3.60606468326431e-06\\
8	8.97e-05	5.89097971832289e-06	5.89097971832289e-06\\
9	0.0002337	9.55645761226858e-06	9.55645761226858e-06\\
10	0.0005997	1.53663918956866e-05	1.53663918956866e-05\\
};
\addlegendentry{Naive MC};
\end{axis}
\end{tikzpicture}%
\caption{Error bars of the MC and IS estimators for the outage probability of $3$-branch MRC receivers over the correlated Nakagami-m fading model. Number of samples $M = 10^7$ and $M^* = 10^4$.}
\label{fig23}
\end{figure}
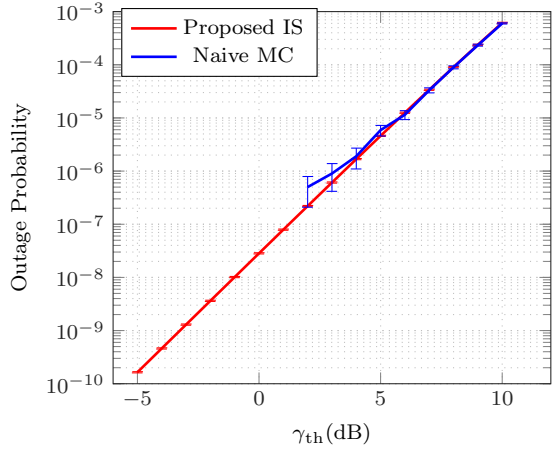
\subsection{Example 3: Maximum Ratio Combining Over Correlated Rician Fading Channel}
\subsubsection{Problem Setup}
In this example, we aim to estimate the outage probability of MRC diversity receivers over correlated Rician fading channels. We assume that the fading at each branch follows a Rice distribution with factor $K_i$. The instantaneous SNR can be expressed as $ \gamma_{end} = \frac{E_s}{N_0} g^{*}g$, where $g$ is a circularly symmetric complex Gaussian random vector with mean $\bar{g}$ and covariance matrix $\Sigma_g$. The expression of the mean and covariance matrix of $g$ can be expressed in terms of the Rician factors as \cite{Pablo2018}
\begin{align}
\bar{g}_i &= \sqrt{\frac{K_i}{1+K_i}}, \\
{\Sigma_g}_{i,j} &= \sqrt{\frac{1}{K_i K_j}} R_{i,j},
\end{align}
where $R = \left(R_{i,j}\right)_{1 \leq i,j \leq L}$ is the correlation matrix of $g$. In this example, we will assume that the correlation matrix has an exponential structure \cite{Karagiannidis2003}, i.e., $R_{i,j} = \rho^{|i-j|}$, $\forall 1 \leq i,j \leq L$.\\
The outage probability, for a given threshold $\gamma_{th} > 0$, is given by
\begin{align}
P_{out} = \mathbb{P}\left(g^{*} g \leq \frac{N_0}{E_s} \gamma_{th} \right).
\end{align}
We can clearly see that this corresponds to the left tail of a complex Gaussian quadratic form with the matrix $\Sigma = I$.
\subsubsection{Results and Discussions}
We start by plotting the outage probability of MRC diversity receivers over the correlated Rician fading model as a function of the threshold $\gamma_{th}$ for different values of $\rho \in \{0.1, 0.5, 0.8\}$ and for two values of the number of branches $L=2$ (Fig. \ref{fig31}) and $L=4$ (Fig. \ref{fig33}). The case $\rho = 0.1$ indicates a low correlation between branches while $\rho = 0.8$ means a high correlation between branches. In these plots, the solid line corresponds to estimates obtained using our proposed IS method, and the dashed one using naive MC. As we can observe from both plots, we can clearly see that although using fewer simulation runs ($10^{-3}$ less), the proposed IS method accurately estimates the outage probability, unlike naive MC, which fails as the probability becomes smaller (below $5\times 10^{-6}$).

\begin{figure}[H]
\centering
\setlength\figureheight{0.3\textwidth}
\setlength\figurewidth{0.35\textwidth}
\scalefont{0.7}
\begin{tikzpicture}

\begin{axis}[%
width=\figurewidth,
height=\figureheight,
at={(0\figurewidth,0\figureheight)},
scale only axis,
every outer x axis line/.append style={darkgray!60!black},
every x tick label/.append style={font=\color{darkgray!60!black}},
xmin=-25,
xmax=-5,
xlabel={$\gamma{}_{\text{th}}\text{(dB)}$},
xmajorgrids,
every outer y axis line/.append style={darkgray!60!black},
every y tick label/.append style={font=\color{darkgray!60!black}},
ymode=log,
ymin=1e-09,
ymax=0.001,
yminorticks=true,
ymajorgrids,
yminorgrids,
grid style={dotted},
ylabel={Outage Probability},
legend style={at={(0.028574218750001,0.784593784825617)},anchor=south west,draw=black,fill=white,align=left}
]
\addplot [
color=green,
solid,
line width=1.0pt,
mark size=1.5pt,
mark=asterisk,
mark options={solid},
]
  table[row sep=crcr]{%
-25	3.62065551303936e-08\\
-24	5.74011092792909e-08\\
-23	9.10083171581555e-08\\
-22	1.44303174338268e-07\\
-21	2.28829729615598e-07\\
-20	3.62911181816724e-07\\
-19	5.75640574074362e-07\\
-18	9.13229887509031e-07\\
-17	1.44911991081966e-06\\
-16	2.30009508025278e-06\\
-15	3.6520016343079e-06\\
-14	5.80085148581937e-06\\
-13	9.21862671250018e-06\\
-12	1.46588290445781e-05\\
-11	2.33261347316994e-05\\
-10	3.71496888319258e-05\\
-9	5.92242936421249e-05\\
-8	9.45239311953324e-05\\
-7	0.000151056984451317\\
-6	0.000241736153288976\\
-5	0.000387398003492278\\
};
\addlegendentry{$\rho = 0.1$};
\addplot [
color=red,
solid,
line width=1.0pt,
mark size=1.5pt,
mark=asterisk,
mark options={solid},
]
  table[row sep=crcr]{%
-25	1.96250752503902e-08\\
-24	3.1109829965078e-08\\
-23	4.93172041358277e-08\\
-22	7.818375938062e-08\\
-21	1.23952510176278e-07\\
-20	1.96525494171041e-07\\
-19	3.1161064176095e-07\\
-18	4.94130432344773e-07\\
-17	7.83635929029134e-07\\
-16	1.24290985400266e-06\\
-15	1.97164459883744e-06\\
-14	3.12820320018282e-06\\
-13	4.96426738760364e-06\\
-12	7.88005707910657e-06\\
-11	1.25124346719951e-05\\
-10	1.98756647505992e-05\\
-9	3.1586644244798e-05\\
-8	5.02259210235159e-05\\
-7	7.99174715350084e-05\\
-6	0.000127261662231251\\
-5	0.000202839489229584\\
};
\addlegendentry{$\rho = 0.5$};
\addplot[
color=blue,
solid,
line width=1.0pt,
mark size=1.5pt,
mark=asterisk,
mark options={solid},
]
  table[row sep=crcr]{%
-25	6.22198914691602e-09\\
-24	9.86430682120029e-09\\
-23	1.56398289575514e-08\\
-22	2.47988691640646e-08\\
-21	3.93254836911971e-08\\
-20	6.23689373100456e-08\\
-19	9.89297254654999e-08\\
-18	1.56951147173737e-07\\
-17	2.49057772551143e-07\\
-16	3.95327316105943e-07\\
-15	6.27716394828448e-07\\
-14	9.97139149553848e-07\\
-13	1.58481276178101e-06\\
-12	2.52048864224303e-06\\
-11	4.01183957033059e-06\\
-10	6.39200700169655e-06\\
-9	1.01968750559061e-05\\
-8	1.62913188714262e-05\\
-7	2.60767310460849e-05\\
-6	4.18345428804044e-05\\
-5	6.72993153471543e-05\\
};
\addlegendentry{$\rho = 0.8$};
\addplot[
color=blue,
dashed,
line width=1.0pt,
mark size=1.5pt,
mark=o,
mark options={solid},
forget plot
]
  table[row sep=crcr]{%
-25	0\\
-24	0\\
-23	0\\
-22	0\\
-21	0\\
-20	0\\
-19	0\\
-18	0\\
-17	1e-07\\
-16	1e-07\\
-15	3e-07\\
-14	8e-07\\
-13	1.4e-06\\
-12	2.1e-06\\
-11	3.9e-06\\
-10	6.6e-06\\
-9	9.8e-06\\
-8	1.68e-05\\
-7	2.61e-05\\
-6	4.29e-05\\
-5	6.78e-05\\
};
\addlegendentry{$\rho = 0.8$};
\addplot [
color=red,
dashed,
line width=1.0pt,
mark size=1.5pt,
mark=o,
mark options={solid},
forget plot
]
  table[row sep=crcr]{%
-25	0\\
-24	0\\
-23	0\\
-22	1e-07\\
-21	3e-07\\
-20	6e-07\\
-19	8e-07\\
-18	1e-06\\
-17	1.3e-06\\
-16	1.8e-06\\
-15	3.2e-06\\
-14	4.3e-06\\
-13	5.8e-06\\
-12	8.2e-06\\
-11	1.22e-05\\
-10	1.99e-05\\
-9	3.23e-05\\
-8	5.1e-05\\
-7	8.22e-05\\
-6	0.000132\\
-5	0.0002062\\
};
\addlegendentry{$\rho = 0.5$};
\addplot [
color=green,
dashed,
line width=1.0pt,
mark size=1.5pt,
mark=o,
mark options={solid},
forget plot
]
  table[row sep=crcr]{%
-25	0\\
-24	2e-07\\
-23	2e-07\\
-22	3e-07\\
-21	5e-07\\
-20	7e-07\\
-19	1.1e-06\\
-18	1.3e-06\\
-17	2.2e-06\\
-16	3.1e-06\\
-15	4.8e-06\\
-14	6.7e-06\\
-13	1.05e-05\\
-12	1.59e-05\\
-11	2.56e-05\\
-10	4.09e-05\\
-9	6.22e-05\\
-8	0.0001001\\
-7	0.0001562\\
-6	0.0002459\\
-5	0.0004017\\
};\end{axis}
\end{tikzpicture}%
\caption{Outage probability of $2$-branch MRC diversity receivers over correlated Rician fading model with $E_s/N_0 = 10$ dB. Solid line corresponds to IS and dashed line to MC. Number of samples $M = 10^7$ and $M^* = 10^4$.}
\label{fig31}
\end{figure}
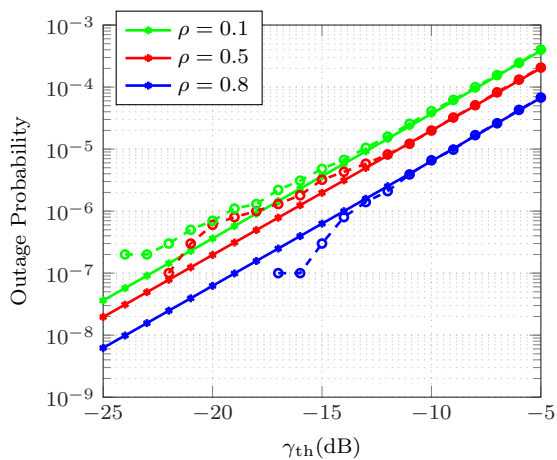
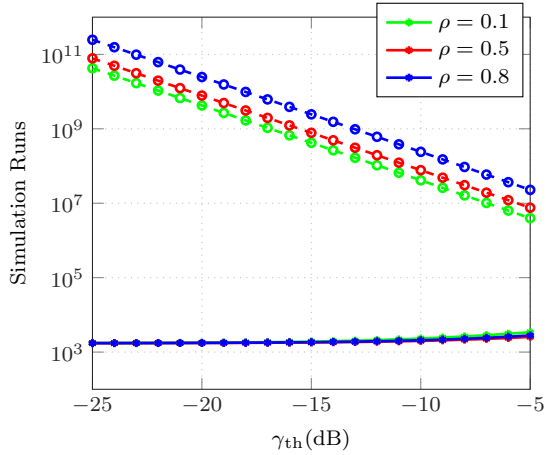
\begin{figure}[H]
\centering
\setlength\figureheight{0.3\textwidth}
\setlength\figurewidth{0.35\textwidth}
\scalefont{0.7}
\begin{tikzpicture}

\begin{axis}[%
width=\figurewidth,
height=\figureheight,
at={(0\figurewidth,0\figureheight)},
scale only axis,
every outer x axis line/.append style={darkgray!60!black},
every x tick label/.append style={font=\color{darkgray!60!black}},
xmin=-25,
xmax=-5,
xlabel={$\gamma{}_{\text{th}}\text{(dB)}$},
xmajorgrids,
every outer y axis line/.append style={darkgray!60!black},
every y tick label/.append style={font=\color{darkgray!60!black}},
ymode=log,
ymin=100,
ymax=1000000000000,
yminorticks=true,
ymajorgrids,
yminorgrids,
grid style={dotted},
ylabel={Simulation Runs},
legend style={at={(0.648574218750001,0.784593784825617)},anchor=south west,draw=black,fill=white,align=left}
]
\addplot [
color=green,
dashed,
line width=1.0pt,
mark size=1.5pt,
mark=o,
mark options={solid},
forget plot
]
  table[row sep=crcr]{%
-25	42440932003.3218\\
-24	26770212824.9557\\
-23	16884609100.9746\\
-22	10648690060.3853\\
-21	6715209824.1449\\
-20	4234202524.82656\\
-19	2669441982.81819\\
-18	1682641597.38117\\
-17	1060393802.99124\\
-16	668075193.401579\\
-15	420764978.23896\\
-14	264897505.122478\\
-13	166687065.461261\\
-12	104825390.212548\\
-11	65874786.9633168\\
-10	41361932.2964998\\
-9	25944572.7266303\\
-8	16255087.270664\\
-7	10171048.2661632\\
-6	6355145.96246978\\
-5	3965030.00760181\\
};
\addplot [
color=green,
solid,
line width=1.0pt,
mark size=1.5pt,
mark=asterisk,
mark options={solid},
]
  table[row sep=crcr]{%
-25	1747.18716264159\\
-24	1752.2355641046\\
-23	1758.58137240632\\
-22	1766.55594124177\\
-21	1776.57408859458\\
-20	1789.15448019323\\
-19	1804.94463522482\\
-18	1824.75140068969\\
-17	1849.57777525085\\
-16	1880.66691338053\\
-15	1919.55394096201\\
-14	1968.12576682932\\
-13	2028.68824198798\\
-12	2104.03861302563\\
-11	2197.53901271354\\
-10	2313.18350395015\\
-9	2455.64682893168\\
-8	2630.29775266716\\
-7	2843.15478914936\\
-6	3100.7598509317\\
-5	3409.95152096314\\
};
\addlegendentry{$\rho = 0.1$};
\addplot [
color=red,
dashed,
line width=1.0pt,
mark size=1.5pt,
mark=o,
mark options={solid},
forget plot
]
  table[row sep=crcr]{%
-25	78299825617.8798\\
-24	49394032494.563\\
-23	31158293563.9468\\
-22	19654208137.7658\\
-21	12397004363.5607\\
-20	7819034901.77048\\
-19	4931280627.91079\\
-18	3109784664.36012\\
-17	1960909063.6483\\
-16	1236323040.76311\\
-15	779368133.180851\\
-14	491219749.723428\\
-13	309538598.896851\\
-12	195002126.985013\\
-11	122807496.156731\\
-10	77311097.6432721\\
-9	48646872.7348918\\
-8	30593024.2697052\\
-7	19226298.90691\\
-6	12073113.1253608\\
-5	7574108.5454439\\
};
\addplot [
color=red,
solid,
line width=1.0pt,
mark size=1.5pt,
mark=asterisk,
mark options={solid},
]
  table[row sep=crcr]{%
-25	1736.98418133821\\
-24	1739.29519224559\\
-23	1742.1887505469\\
-22	1745.81327733157\\
-21	1750.35509732797\\
-20	1756.04807481382\\
-19	1763.18566514174\\
-18	1772.13596797192\\
-17	1783.36049976333\\
-16	1797.43755801294\\
-15	1815.09122878483\\
-14	1837.22729039199\\
-13	1864.9774834794\\
-12	1899.75383805073\\
-11	1943.31494745339\\
-10	1997.84621929556\\
-9	2066.05615576682\\
-8	2151.29054205613\\
-7	2257.66595882734\\
-6	2390.22320487141\\
-5	2555.10002841077\\
};
\addlegendentry{$\rho = 0.5$};
\addplot [
color=blue,
dashed,
line width=1.0pt,
mark size=1.5pt,
mark=o,
mark options={solid},
forget plot
]
  table[row sep=crcr]{%
-25	246969249568.796\\
-24	155777796929.388\\
-23	98251712351.7051\\
-22	61964114239.5223\\
-21	39074915178.09\\
-20	24637904226.6268\\
-19	15532640374.2675\\
-18	9790560862.3657\\
-17	6169811933.78464\\
-16	3887005349.54288\\
-15	2447982955.4967\\
-14	1541047173.24923\\
-13	969601962.938793\\
-12	609658024.70304\\
-11	383024746.704968\\
-10	240398700.655133\\
-9	150695612.395867\\
-8	94321090.7744787\\
-7	58926094.9440482\\
-6	36729831.6073575\\
-5	22831385.09291\\
};
\addplot [
color=blue,
solid,
line width=1.0pt,
mark size=1.5pt,
mark=asterisk,
mark options={solid},
]
  table[row sep=crcr]{%
-25	1739.31081024013\\
-24	1742.25395446698\\
-23	1745.94656779335\\
-22	1750.5805673313\\
-21	1756.39704045089\\
-20	1763.69879776317\\
-19	1772.86610290342\\
-18	1784.37636863003\\
-17	1798.82879743405\\
-16	1816.9751744548\\
-15	1839.75829731013\\
-14	1868.3598569285\\
-13	1904.2599684778\\
-12	1949.31098901098\\
-11	2005.82873414511\\
-10	2076.70468450518\\
-9	2165.54318154457\\
-8	2276.82781708594\\
-7	2416.12098603308\\
-6	2590.29950551606\\
-5	2807.82668792634\\
};
\addlegendentry{$\rho = 0.8$};
\end{axis}
\end{tikzpicture}%
\caption{Number of required simulation runs for $5\%$ relative error for $2$-branch MRC diversity receivers over correlated Rician fading model with $E_s/N_0 = 10$ dB. Solid line corresponds to IS and dashed line to MC.}
\label{fig32}
\end{figure}
For both numbers of branches $L \in \{2, 4\}$, we plot the required number of simulation runs for a $5\%$ accuracy requirement in Fig. \ref{fig32} and Fig. \ref{fig34}, respectively. From these plots, we can confirm again that the proposed IS scheme outperforms naive MC. In fact, for a fixed threshold $\gamma_{th}$, we notice that the number of simulation runs of IS is far less than the one needed by naive MC to achieve the same accuracy. We also note that while the number of samples of naive MC continues to increase at a high rate as the probability becomes smaller, the number of samples required by the IS estimator remains almost constant as a result of the bounded relative error property. 
\begin{figure}[H]
\centering
\setlength\figureheight{0.3\textwidth}
\setlength\figurewidth{0.35\textwidth}
\scalefont{0.7}
\begin{tikzpicture}

\begin{axis}[%
width=\figurewidth,
height=\figureheight,
at={(0\figurewidth,0\figureheight)},
scale only axis,
every outer x axis line/.append style={darkgray!60!black},
every x tick label/.append style={font=\color{darkgray!60!black}},
xmin=-10,
xmax=5,
xlabel={$\gamma{}_{\text{th}}\text{(dB)}$},
xmajorgrids,
every outer y axis line/.append style={darkgray!60!black},
every y tick label/.append style={font=\color{darkgray!60!black}},
ymode=log,
ymin=1e-12,
ymax=0.01,
yminorticks=true,
ymajorgrids,
yminorgrids,
grid style={dotted},
ylabel={Outage Probability},
legend style={at={(0.028574218750001,0.784593784825617)},anchor=south west,draw=black,fill=white,align=left}
]
\addplot [
color=green,
dashed,
line width=1.0pt,
mark size=1.5pt,
mark=asterisk,
mark options={solid},
forget plot
]
  table[row sep=crcr]{%
-10	0\\
-9	1e-07\\
-8	2e-07\\
-7	2e-07\\
-6	2e-07\\
-5	9e-07\\
-4	2.3e-06\\
-3	5.2e-06\\
-2	1.12e-05\\
-1	2.61e-05\\
0	6.24e-05\\
1	0.0001463\\
2	0.0003277\\
3	0.0007136\\
4	0.0014935\\
5	0.0030265\\
};
\addplot [
color=green,
solid,
line width=1.0pt,
mark size=1.5pt,
mark=asterisk,
mark options={solid},
]
  table[row sep=crcr]{%
-10	9.65714964963335e-09\\
-9	2.41237904179929e-08\\
-8	6.01699376120886e-08\\
-7	1.49782313453461e-07\\
-6	3.71911663120003e-07\\
-5	9.20432342913409e-07\\
-4	2.26824647732407e-06\\
-3	5.55874315718557e-06\\
-2	1.12e-05\\
-1	2.61e-05\\
0	6.24e-05\\
1	0.0001463\\
2	0.0003277\\
3	0.0007136\\
4	0.0014935\\
5	0.0030265\\
};
\addlegendentry{$\rho = 0.1$};
\addplot[
color=red,
dashed,
line width=1.0pt,
mark size=1.5pt,
mark=asterisk,
mark options={solid},
forget plot
]
  table[row sep=crcr]{%
-10	0\\
-9	0\\
-8	0\\
-7	0\\
-6	0\\
-5	2e-07\\
-4	4e-07\\
-3	5e-07\\
-2	1.5e-06\\
-1	3e-06\\
0	6.9e-06\\
1	1.69e-05\\
2	4.05e-05\\
3	9.25e-05\\
4	0.0002094\\
5	0.0004745\\
};
\addplot[
color=red,
solid,
line width=1.0pt,
mark size=1.5pt,
mark=o,
mark options={solid},
]
  table[row sep=crcr]{%
-10	7.78373293023172e-10\\
-9	1.94618827660978e-09\\
-8	4.86055982389739e-09\\
-7	1.21218273324297e-08\\
-6	3.01769418247368e-08\\
-5	7.49573799201945e-08\\
-4	1.85671432385381e-07\\
-3	4.58319971790844e-07\\
-2	1.12646124975491e-06\\
-1	2.75380291401249e-06\\
0	6.68747201720957e-06\\
1	1.61074721087067e-05\\
2	3.84074656120456e-05\\
3	9.04609223249168e-05\\
4	0.000209911580707446\\
5	0.00047847112816254\\
};
\addlegendentry{$\rho = 0.5$};
\addplot [
color=blue,
dashed,
line width=1.0pt,
mark size=1.5pt,
mark=o,
mark options={solid},
forget plot
]
  table[row sep=crcr]{%
-10	0\\
-9	0\\
-8	0\\
-7	0\\
-6	0\\
-5	1e-07\\
-4	1e-07\\
-3	1e-07\\
-2	1e-07\\
-1	4e-07\\
0	9e-07\\
1	2.1e-06\\
2	4.2e-06\\
3	7.9e-06\\
4	2.29e-05\\
5	5.15e-05\\
};
\addplot [
color=blue,
solid,
line width=1.0pt,
mark size=1.5pt,
mark=o,
mark options={solid},
]
  table[row sep=crcr]{%
-10	4.95429126626755e-11\\
-9	1.24311550417019e-10\\
-8	3.11860502196412e-10\\
-7	7.8218883042615e-10\\
-6	1.96130199415539e-09\\
-5	4.91623742092375e-09\\
-4	1.23180942953332e-08\\
-3	3.08484793762464e-08\\
-2	7.720544738564e-08\\
-1	1.93069133927769e-07\\
0	4.82315210322479e-07\\
1	1.20329202855523e-06\\
2	2.99678044080943e-06\\
3	7.44627686036533e-06\\
4	1.84454859347364e-05\\
5	4.55040943554115e-05\\
};
\addlegendentry{$\rho = 0.8$};
\end{axis}
\end{tikzpicture}%
\caption{Outage probability of $4$-branch MRC diversity receivers over the correlated Rician fading model with $E_s/N_0 = 10$ dB. Solid line corresponds to IS and dashed line to MC. Number of samples $M = 10^7$ and $M^* = 10^4$.}
\label{fig33}
\end{figure}
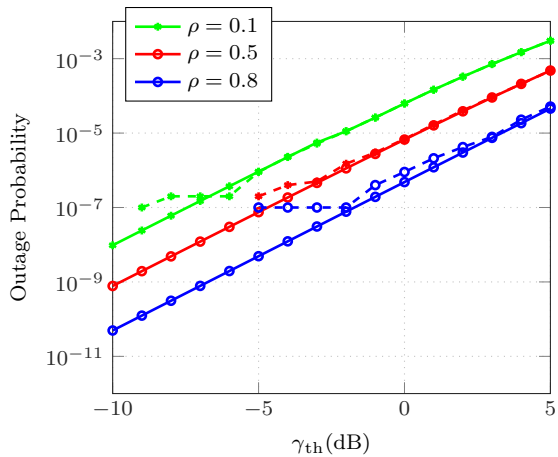
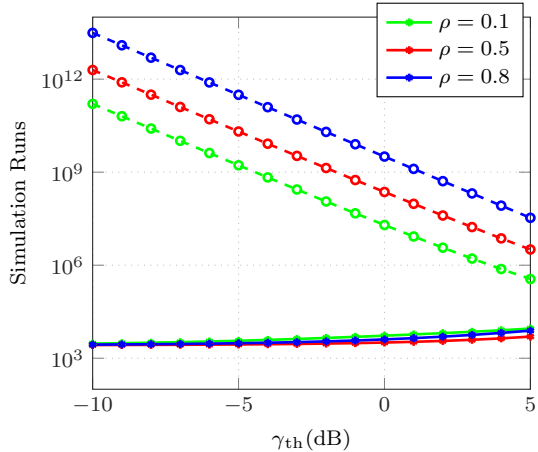
\begin{figure}[H]
\centering
\setlength\figureheight{0.3\textwidth}
\setlength\figurewidth{0.35\textwidth}
\scalefont{0.7}
\begin{tikzpicture}

\begin{axis}[%
width=\figurewidth,
height=\figureheight,
at={(0\figurewidth,0\figureheight)},
scale only axis,
every outer x axis line/.append style={darkgray!60!black},
every x tick label/.append style={font=\color{darkgray!60!black}},
xmin=-10,
xmax=5,
xlabel={$\gamma{}_{\text{th}}\text{(dB)}$},
xmajorgrids,
every outer y axis line/.append style={darkgray!60!black},
every y tick label/.append style={font=\color{darkgray!60!black}},
ymode=log,
ymin=100,
ymax=100000000000000,
yminorticks=true,
ymajorgrids,
yminorgrids,
grid style={dotted},
ylabel={Simulation Runs},
legend style={at={(0.648574218750001,0.784593784825617)},anchor=south west,draw=black,fill=white,align=left}
]
\addplot [
color=green,
dashed,
line width=1.0pt,
mark size=1.5pt,
mark=o,
mark options={solid},
forget plot
]
  table[row sep=crcr]{%
-10	159119413171.647\\
-9	63698114446.5716\\
-8	25538333070.0304\\
-7	10259153663.7999\\
-6	4131732292.59508\\
-5	1669474782.64723\\
-4	677455704.167811\\
-3	276435052.809837\\
-2	113616968.475483\\
-1	47137852.8735796\\
0	19795589.3285916\\
1	8443656.8555216\\
2	3673360.07796469\\
3	1637843.60189043\\
4	752460.895008273\\
5	358174.979666907\\
};
\addplot [
color=green,
solid,
line width=1.0pt,
mark size=1.5pt,
mark=asterisk,
mark options={solid},
]
  table[row sep=crcr]{%
-10	2970.46610902712\\
-9	3058.39646588934\\
-8	3166.96652327147\\
-7	3300.32261812477\\
-6	3463.09220078246\\
-5	3660.27059045655\\
-4	3897.03434275038\\
-3	4178.5065679988\\
-2	4509.56597369524\\
-1	4894.90251262726\\
0	5339.66973740077\\
1	5851.22207344233\\
2	6442.48956224119\\
3	7137.59069682021\\
4	7980.83378291249\\
5	9052.58689072958\\
};
\addlegendentry{$\rho = 0.1$};
\addplot [
color=red,
dashed,
line width=1.0pt,
mark size=1.5pt,
mark=o,
mark options={solid},
forget plot
]
  table[row sep=crcr]{%
-10	1974168451792.16\\
-9	789563895475.83\\
-8	316144651687.248\\
-7	126766364445.905\\
-6	50920996652.1286\\
-5	20500181389.1243\\
-4	8276123553.03204\\
-3	3352765295.65778\\
-2	1364128832.10045\\
-1	558004990.327105\\
0	229777368.758872\\
1	95397666.2659227\\
2	40007351.619424\\
3	16985245.7242203\\
4	7318878.91221108\\
5	3210026.00099177\\
};
\addplot [
color=red,
solid,
line width=1.0pt,
mark size=1.5pt,
mark=asterisk,
mark options={solid},
]
  table[row sep=crcr]{%
-10	2663.74961740967\\
-9	2677.40510924847\\
-8	2694.93761682316\\
-7	2717.42870676342\\
-6	2746.27288984714\\
-5	2783.27849376537\\
-4	2830.80825684745\\
-3	2891.97980533443\\
-2	2970.95879119075\\
-1	3073.39888631391\\
0	3207.11929525087\\
1	3383.17230415246\\
2	3617.55716766011\\
3	3934.00737334906\\
4	4368.55033746947\\
5	4976.94928464954\\
};
\addlegendentry{$\rho = 0.5$};
\addplot [
color=blue,
dashed,
line width=1.0pt,
mark size=1.5pt,
mark=o,
mark options={solid},
forget plot
]
  table[row sep=crcr]{%
-10	31016343556270.9\\
-9	12361200505135\\
-8	4927331254513.9\\
-7	1964538406871.48\\
-6	783479546528.439\\
-5	312564235792.498\\
-4	124746567466.502\\
-3	49812502388.0502\\
-2	19903257262.2438\\
-1	7959012774.6527\\
0	3185964750.78559\\
1	1277028447.38227\\
2	512762087.639713\\
3	206361996.27927\\
4	83305566.541609\\
5	33767732.2965052\\
};
\addplot [
color=blue,
solid,
line width=1.0pt,
mark size=1.5pt,
mark=asterisk,
mark options={solid},
]
  table[row sep=crcr]{%
-10	2747.92654922765\\
-9	2782.81403224219\\
-8	2827.0748018036\\
-7	2883.24852169054\\
-6	2954.58770329731\\
-5	3045.27368977079\\
-4	3160.70724480749\\
-3	3307.90432954139\\
-2	3496.04236260515\\
-1	3737.22484255769\\
0	4047.56715924534\\
1	4448.76105327719\\
2	4970.36202350412\\
3	5653.18615654635\\
4	6554.44735460373\\
5	7755.71700139074\\
};
\addlegendentry{$\rho = 0.8$};
\end{axis}
\end{tikzpicture}%
\caption{Number of required simulation runs for $5\%$ relative error for $4$-branch MRC diversity receivers over the correlated Rician fading model with  $E_s/N_0 = 10$ dB. Solid line corresponds to IS and dashed line to MC.}
\label{fig34}
\end{figure}

In Fig. \ref{fig35}, we plot the outage probability of 2-and 4-branch MRC receivers and their error bars to determine the relative error of both methods. We can clearly see that the magnitude of the error bars is larger for naive MC than for the proposed IS scheme, which indicates that the relative error of naive MC tends to increase despite using more simulation runs compared to our IS estimator.

\begin{figure}[H]
\centering
\setlength\figureheight{0.3\textwidth}
\setlength\figurewidth{0.35\textwidth}
\scalefont{0.7}
\begin{tikzpicture}

\begin{axis}[%
width=\figurewidth,
height=\figureheight,
at={(0\figurewidth,0\figureheight)},
scale only axis,
every outer x axis line/.append style={darkgray!60!black},
every x tick label/.append style={font=\color{darkgray!60!black}},
xmin=-30,
xmax=6,
xmajorgrids,
xlabel={$\gamma{}_{\text{th}}\text{(dB)}$},
xmajorgrids,
every outer y axis line/.append style={darkgray!60!black},
every y tick label/.append style={font=\color{darkgray!60!black}},
ymode=log,
ymin=1e-10,
ymax=0.001,
yminorticks=true,
ymajorgrids,
yminorgrids,
grid style={dotted},
ylabel={Outage Probability},
legend style={at={(0.018574218750001,0.824593784825617)},anchor=south west,draw=black,fill=white,align=left}
]
\addplot [
color=red,
solid,
line width=1.0pt,
]
 plot [error bars/.cd, y dir = both, y explicit]
 table[row sep=crcr, y error plus index=2, y error minus index=3]{%
-25	1.96250752503902e-08	4.08958525875024e-10	4.08958525875024e-10\\
-24	3.1109829965078e-08	6.48715522683297e-10	6.48715522683297e-10\\
-23	4.93172041358277e-08	1.02923857476915e-09	1.02923857476915e-09\\
-22	7.818375938062e-08	1.63337329041999e-09	1.63337329041999e-09\\
-21	1.23952510176278e-07	2.59291578914625e-09	2.59291578914625e-09\\
-20	1.96525494171041e-07	4.11772275144706e-09	4.11772275144706e-09\\
-19	3.1161064176095e-07	6.5423128537625e-09	6.5423128537625e-09\\
-18	4.94130432344773e-07	1.04006414868506e-08	1.04006414868506e-08\\
-17	7.83635929029134e-07	1.65464151154741e-08	1.65464151154741e-08\\
-16	1.24290985400266e-06	2.63473257899105e-08	2.63473257899105e-08\\
-15	1.97164459883744e-06	4.19998623240641e-08	4.19998623240641e-08\\
-14	3.12820320018282e-06	6.70419139438942e-08	6.70419139438942e-08\\
-13	4.96426738760364e-06	1.07191899357318e-07	1.07191899357318e-07\\
-12	7.88005707910657e-06	1.71730736680257e-07	1.71730736680257e-07\\
-11	1.25124346719951e-05	2.75793116808548e-07	2.75793116808548e-07\\
-10	1.98756647505992e-05	4.4419400733611e-07	4.4419400733611e-07\\
-9	3.1586644244798e-05	7.17867933371791e-07	7.17867933371791e-07\\
-8	5.02259210235159e-05	1.16478950177577e-06	1.16478950177577e-06\\
-7	7.99174715350084e-05	1.89863543369785e-06	1.89863543369785e-06\\
-6	0.000127261662231251	3.11090554468749e-06	3.11090554468749e-06\\
-5	0.000202839489229584	5.12656497431674e-06	5.12656497431674e-06\\
};
\addlegendentry{Proposed IS};
\addplot [
color=blue,
solid,
line width=1.0pt,
]
 plot [error bars/.cd, y dir = both, y explicit]
 table[row sep=crcr, y error plus index=2, y error minus index=3]{%
-25	0	8.68283868342212e-08	8.68283868342212e-08\\
-24	0	1.09321324121079e-07	1.09321324121079e-07\\
-23	0	1.37643366009666e-07	1.37643366009666e-07\\
-22	0	1.73306291447841e-07	1.73306291447841e-07\\
-21	3e-07	2.18214551318188e-07	2.18214551318188e-07\\
-20	6e-07	2.74767572693024e-07	2.74767572693024e-07\\
-19	8e-07	3.45988882533085e-07	3.45988882533085e-07\\
-18	1e-06	4.35689170270717e-07	4.35689170270717e-07\\
-17	1.3e-06	5.48672345383686e-07	5.48672345383686e-07\\
-16	1.8e-06	6.90996133168431e-07	6.90996133168431e-07\\
-15	3.2e-06	8.70301956629168e-07	8.70301956629168e-07\\
-14	4.3e-06	1.0962329962763e-06	1.0962329962763e-06\\
-13	5.8e-06	1.38096469629086e-06	1.38096469629086e-06\\
-12	8.2e-06	1.73987898227989e-06	1.73987898227989e-06\\
-11	1.22e-05	2.19242257767787e-06	2.19242257767787e-06\\
-10	1.99e-05	2.76320169572121e-06	2.76320169572121e-06\\
-9	3.23e-05	3.48338082478311e-06	3.48338082478311e-06\\
-8	5.1e-05	4.39247318965115e-06	4.39247318965115e-06\\
-7	8.22e-05	5.54063555117406e-06	5.54063555117406e-06\\
-6	0.000132	6.99161057895059e-06	6.99161057895059e-06\\
-5	0.0002062	8.82649490782143e-06	8.82649490782143e-06\\
};
\addlegendentry{Naive MC};
\addplot [
color=red,
solid,
line width=1.0pt,
forget plot
]
 plot [error bars/.cd, y dir = both, y explicit]
 table[row sep=crcr, y error plus index=2, y error minus index=3]{%
-10	7.78373293023172e-10	2.0086516730417e-11	2.0086516730417e-11\\
-9	1.94618827660978e-09	5.03514405717795e-11	5.03514405717795e-11\\
-8	4.86055982389739e-09	1.26162607228114e-10	1.26162607228114e-10\\
-7	1.21218273324297e-08	3.15949139793912e-10	3.15949139793912e-10\\
-6	3.01769418247368e-08	7.90709715141959e-10	7.90709715141959e-10\\
-5	7.49573799201945e-08	1.97725523456423e-09	1.97725523456423e-09\\
-4	1.85671432385381e-07	4.93935588983141e-09	4.93935588983141e-09\\
-3	4.58319971790844e-07	1.23235665532109e-08	1.23235665532109e-08\\
-2	1.12646124975491e-06	3.06997313410696e-08	3.06997313410696e-08\\
-1	2.75380291401249e-06	7.63329987947856e-08	7.63329987947856e-08\\
0	6.68747201720957e-06	1.89360564445709e-07	1.89360564445709e-07\\
1	1.61074721087067e-05	4.68445909240231e-07	4.68445909240231e-07\\
2	3.84074656120456e-05	1.15503024394919e-06	1.15503024394919e-06\\
3	9.04609223249168e-05	2.8369298735627e-06	2.8369298735627e-06\\
4	0.000209911580707446	6.93705400549952e-06	6.93705400549952e-06\\
5	0.00047847112816254	1.68774701573947e-05	1.68774701573947e-05\\
};
\addplot [
color=blue,
solid,
line width=1.0pt,
forget plot
]
 plot [error bars/.cd, y dir = both, y explicit]
 table[row sep=crcr, y error plus index=2, y error minus index=3]{%
-10	0	1.7292191417372e-08	1.7292191417372e-08\\
-9	0	2.7343146982148e-08	2.7343146982148e-08\\
-8	0	4.32114875105292e-08	4.32114875105292e-08\\
-7	0	6.82401724175594e-08	6.82401724175594e-08\\
-6	0	1.07669743296603e-07	1.07669743296603e-07\\
-5	2e-07	1.69692736767656e-07	1.69692736767656e-07\\
-4	4e-07	2.67072132993471e-07	2.67072132993471e-07\\
-3	5e-07	4.19604718357168e-07	4.19604718357168e-07\\
-2	1.5e-06	6.57830423619554e-07	6.57830423619554e-07\\
-1	3e-06	1.02854169297913e-06	1.02854169297913e-06\\
0	6.9e-06	1.60282315606538e-06	1.60282315606538e-06\\
1	1.69e-05	2.48751820389642e-06	2.48751820389642e-06\\
2	4.05e-05	3.84109949132678e-06	3.84109949132678e-06\\
3	9.25e-05	5.89477092604113e-06	5.89477092604113e-06\\
4	0.0002094	8.97901473725121e-06	8.97901473725121e-06\\
5	0.0004745	1.35543912099037e-05	1.35543912099037e-05\\
};
\node[black,draw,below] at (axis cs:-20.5,1.0e-05){{$L = 2$}};
\node[black,draw,below] at (axis cs:0,1.0e-08){{$L = 4$}};
\end{axis}
\end{tikzpicture}%
\caption{Error bars of MC and IS estimators of the outage probability of $L$-branch MRC receivers over the correlated Rician fading model. Number of samples $M = 10^7$ and $M^* = 10^4$.}
\label{fig35}
\end{figure}
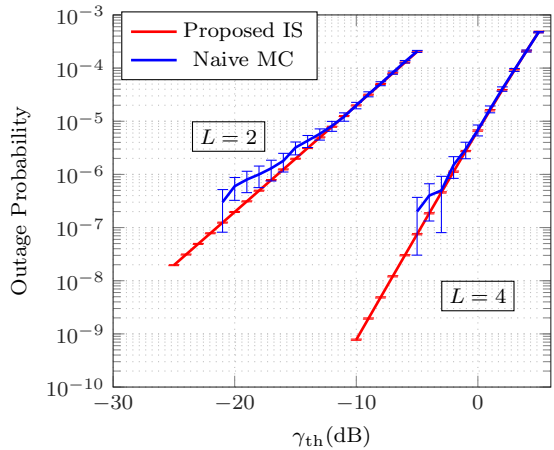
\section{Conclusion}
In this paper, we proposed efficient IS estimators for the left tail of the positive quadratic form in Gaussian random vectors. We discussed the construction of these estimators in both central and non-central cases, as well as both real and complex settings. We showed that these estimators are endowed with the bounded relative error property, making them an appealing alternative to naive MC, especially when the probability of interest is very small. To validate our results, we presented three examples, a toy example and two examples motivated by wireless communication theory. A clear gain in terms of simulation runs was observed from the numerical simulations, confirming the efficiency of our scheme compared to naive MC. However, we should note that the upper bound of the relative error of the proposed IS scheme became loose as the dimension of the problem at hand $N$ becomes larger, in both the real-valued and complex-valued cases.

\nocite{*}
\bibliographystyle{unsrt} 
\bibliography{reference}

\begin{thebibliography}{10}
\expandafter\ifx\csname url\endcsname\relax
  \def\url#1{\texttt{#1}}\fi
\expandafter\ifx\csname urlprefix\endcsname\relax\def\urlprefix{URL }\fi
\expandafter\ifx\csname href\endcsname\relax
  \def\href#1#2{#2} \def\path#1{#1}\fi

\bibitem{box1954}
G.~E.~P. Box, Some theorems on quadratic forms applied in the study of analysis
  of variance problems, ii. effects of inequality of variance and of
  correlation between errors in the two-way classification, Ann. Math. Statist.
  25~(3) (1954) 484--498.

\bibitem{HSUAN1985}
F.~Hsuan, P.~Langenberg, A.~Getson, The {2}-inverse with applications in
  statistics, Linear Algebra and its Applications 70 (1985) 241--248.

\bibitem{Rice1980}
S.~O. Rice, Distribution of quadratic forms in normal random
  variables---evaluation by numerical integration, SIAM J. Sci. Stat. Comput. 1
  (1980) 438--448.

\bibitem{Shapiro1983}
A.~Shapiro, Asymptotic distribution theory in the analysis of covariance
  structures, South African Statistical Journal 17~(1) (1983) 33--81.

\bibitem{Vuong1989}
Q.~H. Vuong, Likelihood ratio tests for model selection and non-nested
  hypotheses, Econometrica 57~(2) (1989) 307--333.

\bibitem{Tang2010}
L.~Tong, J.~Yang, R.~S. Cooper, Efficient calculation of p-value and power for
  quadratic form statistics in multilocus association testing, Annals of Human
  Genetics 74~(3) (2010) 275--285.

\bibitem{Bausch2013}
J.~Bausch, On the efficient calculation of a linear combination of {C}hi-square
  random variables with an application in counting string vacua, Journal of
  Physics A: Mathematical and Theoretical 46~(50) (2013) 505202.

\bibitem{Divsalar1990}
D.~Divsalar, M.~K. Simon, M.~Shahshahani, The performance of trellis-coded
  mdpsk with multiple symbol detection, IEEE Transactions on Communications
  38~(9) (1990) 1391--1403.

\bibitem{Raphaeli1996}
D.~Raphaeli, Noncoherent coded modulation, IEEE Transactions on Communications
  44~(2) (1996) 172--183.

\bibitem{Proakis1994}
J.~G. Proakis, Digital Communications, 3rd Edition, New York, NY: McGraw Hill,
  1994.

\bibitem{QPM}
S.~B. Provost, A.~M. Mathai, Quadratic Forms in Random Variables: Theory and
  Applications, Statistics : Textbooks and Monographs, Marcel Dek ker, 1992.

\bibitem{TSAY}
S.~Andr\`as, A.~Baricz, Y.~Sun, The generalized {M}arcum-{Q} function: {A}n
  orthogonal polynomial approach, Acta Univ. Sapientiae Math. 3~(1) (2011)
  60--76.

\bibitem{TII}
I.~S. Gradshteyn, I.~M. Ryzhik, Table of Integrals, Series, and Products,
  {S}eventh {E}dition Edition, {E}lsevier/{A}cademic {P}ress, Amsterdam, 2007.

\bibitem{gurland1953}
J.~Gurland, Distribution of quadratic forms and ratios of quadratic forms, Ann.
  Math. Statist. 24~(3) (1953) 416--427.

\bibitem{gurland1955}
J.~Gurland, Distribution of definite and of indefinite quadratic forms, Ann.
  Math. Statist. 26~(1) (1955) 122--127.

\bibitem{ruben1960}
H.~Ruben, Probability content of regions under spherical normal distributions,
  {I}, Ann. Math. Statist. 31~(3) (1960) 598--618.

\bibitem{ruben1962}
H.~Ruben, Probability content of regions under spherical normal distributions,
  {IV}: The distribution of homogeneous and non-homogeneous quadratic functions
  of normal variables, Ann. Math. Statist. 33~(2) (1962) 542--570.

\bibitem{shah1961}
B.~K. Shah, C.~G. Khatri, Distribution of a definite quadratic form for
  non-central normal variates, Ann. Math. Statist. 32~(3) (1961) 883--887.

\bibitem{shah1963}
B.~K. Shah, Distribution of definite and of indefinite quadratic forms from a
  non-central normal distribution, Ann. Math. Statist. 34~(1) (1963) 186--190.

\bibitem{kotz1967a}
S.~Kotz, N.~L. Johnson, D.~W. Boyd, Series representations of distributions of
  quadratic forms in normal variables. {I}. {C}entral case, Ann. Math. Statist.
  38~(3) (1967) 823--837.

\bibitem{kotz1967b}
S.~Kotz, N.~L. Johnson, D.~W. Boyd, Series representations of distributions of
  quadratic forms in normal variables. {II} . {N}on-{C}entral case, Ann. Math.
  Statist. 38~(3) (1967) 838--848.

\bibitem{imhof1961}
J.~P. Imhof, Computing the distribution of quadratic forms in normal variables,
  Biometrika 48 (1961) 419--426.

\bibitem{davies1963}
R.~B. Davies, Numerical inversion of a characteristic function, Biometrika
  60~(2) (1973) 415--417.

\bibitem{solomon1977}
H.~Solomon, M.~A. Stephens, Distribution of a sum of weighted chi-square
  variables, Journal of the American Statistical Association 72 (1977)
  881--885.

\bibitem{Shi2016}
Y.~Shi, H.~Kang, H.~Jiang, Efficiently estimating small p-values in permutation
  tests using importance sampling and cross-entropy method arXiv:1608.00053.
\newblock \href {http://arxiv.org/abs/arXiv:1608.00053}
  {\path{arXiv:arXiv:1608.00053}}.

\bibitem{Shi2018}
Y.~Shi, M.~Wang, W.~Shi, J.~Lee, H.~Kang, H.~Jiang, Accurate and efficient
  estimation of small p-values with the cross-entropy method: Applications in
  genomic data analysis arXiv:1803.03373.
\newblock \href {http://arxiv.org/abs/arXiv:1803.03373}
  {\path{arXiv:arXiv:1803.03373}}.

\bibitem{Karagiannidis2003}
G.~K. Karagiannidis, D.~A. Zogas, S.~A. Kotsopoulos, On the multivariate
  {N}akagami-m distribution with exponential correlation, IEEE Transactions on
  Communications 51~(8) (2003) 1240--1244.

\bibitem{Pablo2018}
P.~Ram{\'{\i}}rez{-}Espinosa, L.~Moreno{-}Pozas, J.~F. Paris, J.~A.
  Cort{\'{e}}s, E.~Martos{-}Naya, A new approach to the statistical analysis of
  non-central complex {G}aussian quadratic forms with applications
  arXiv:1805.09181.
\newblock \href {http://arxiv.org/abs/1805.09181} {\path{arXiv:1805.09181}}.

\bibitem{Slichenko2014}
M.~P. Slichenko, Characteristic function of a quadratic form formed by
  correlated complex {G}aussian variables, Journal of Communications Technology
  and Electronics 59~(5) (2014) 433--440.

\bibitem{Al-Naffouri2016}
T.~Y. Al-Naffouri, M.~Moinuddin, N.~Ajeeb, B.~Hassibi, A.~L. Moustakas, On the
  distribution of indefinite quadratic forms in {G}aussian random variables,
  IEEE Transactions on Communications 64~(1) (2016) 153--165.

\bibitem{Turin1960}
G.~L. Turin, The characteristic function of {H}ermitian quadratic forms in
  complex normal variables, Biometrika 47 (1960) 199--201.

\bibitem{Biyari1993}
K.~H. Biyari, W.~C. Lindsey, Statistical distributions of {H}ermitian quadratic
  forms in complex {G}aussian variables, IEEE Transactions on Information
  Theory 39~(3) (1993) 1076--1082.

\bibitem{Holm2004}
H.~Holm, M.-S. Alouini, Sum and difference of two squared correlated nakagami
  variates in connection with the mckay distribution, IEEE Transactions on
  Communications 52~(8) (2004) 1367--1376.

\bibitem{Simon2001}
M.~K. Simon, M.-S. Alouini, On the difference of two chi-square variates with
  application to outage probability computation, IEEE Transactions on
  Communications 49~(11) (2001) 1946--1954.

\end{thebibliography}

\end{document}